\documentclass[letterpaper, 10 pt, conference]{ieeeconf}
\IEEEoverridecommandlockouts             
\usepackage{graphics} 
\usepackage{epsfig} 
\usepackage{times} 
\usepackage{amsmath} 
\usepackage{amssymb}  
\usepackage{xcolor}
\usepackage{nccmath}
\usepackage{hyperref}
\usepackage{fancyhdr}
\usepackage{bm}
\usepackage{dblfloatfix}
\usepackage{mathtools, cuted}
\usepackage{cancel}
\usepackage{placeins}
\usepackage{color}
\usepackage{lipsum}
\let\labelindent\relax\usepackage{enumitem}
\usepackage{setspace}
\usepackage{balance}
\usepackage{wrapfig}
\usepackage{tikz}
\usetikzlibrary{shapes,arrows.meta,calc}
\usetikzlibrary{positioning}
\usepackage{subcaption}
\usepackage{booktabs}
\usepackage{braket}
\usepackage{multirow}
\usepackage{multicol}
\let\labelindent\relax
\usepackage{enumerate}   
\usepackage{soul}
\usepackage{cite}
\usepackage{algorithm} 
\usepackage{algorithmic} 
\usepackage[normalem]{ulem}
\usepackage{xcolor}

\newcommand{\ellone}{$\mathcal{L}_1$ }

\let\labelindent\relax
\usepackage{enumitem}
\usepackage{tabularx}

\definecolor{lcolor}{rgb}{0.79,0.88,1.0}
\definecolor{bcolor}{RGB}{252,184,203}
\definecolor{ccolor}{RGB}{234,208,255}
\tikzstyle{sqblock} = [draw, fill=bcolor!20, rectangle, 
    minimum height=2em, minimum width=3em]
\tikzstyle{block} = [draw, fill=bcolor!20, rectangle, 
    minimum height=2em, minimum width=3em]
\tikzstyle{sum} = [draw, fill=bcolor!20, circle, node distance=1cm]
\tikzstyle{input} = [coordinate]
\tikzstyle{output} = [coordinate]
\tikzstyle{phantom} = [coordinate]

\newcommand{\dd}{\text{d}}
\newcommand{\diag}[1]{\text{diag}(#1)}
\newcommand{\norm}[1]{\left\| #1 \right\|}
\newcommand{\abs}[1]{\lvert #1 \rvert}

\newtheorem{assumption}{Assumption}
\newtheorem{proposition}{Proposition}

\newtheorem{remark}{Remark}
\newtheorem{theorem}{Theorem}

\DeclareMathOperator{\vect}{vec}

\makeatletter
\newcommand*\titleheader[1]{\gdef\@titleheader{#1}}
\AtBeginDocument{%
  \let\st@red@title\@title
  \def\@title{%
  \vspace{-4em}\bgroup\normalfont\large\centering\@titleheader\par\egroup
    \vskip2.5em\st@red@title}
}
\makeatother

\title{\LARGE \bf
\bm{$\mathcal{L}_1$}Quad: 
\bm{$\mathcal{L}_1$} Adaptive Augmentation of Geometric Control for Agile Quadrotors with Performance Guarantees}
\titleheader{This article has been accepted for publication in IEEE Transactions on Control Systems Technology.}

\author{Zhuohuan Wu$^\dagger$, Sheng Cheng$^\dagger$, Pan Zhao, Aditya Gahlawat, Kasey A. Ackerman,  \\ Arun Lakshmanan, Chengyu Yang, Jiahao Yu and Naira Hovakimyan
\thanks{*This work is supported by National Aeronautics and Space Administration (NASA) grant 80NSSC22M0070, National Science Foundation (NSF) under the RI grant \#2133656 and Air Force Office of Scientific Research.}
\thanks{$^\dagger$These authors contributed equally to this work.}
\thanks{Z.~Wu, S. Cheng, A. Gahlawat, K. A.~Ackerman, A.~Lakshmanan, C.~Yang, J.~Yu, and N. Hovakimyan are affiliated with
the Department of Mechanical Science and Engineering, University of Illinois Urbana-Champaign, IL 61801, USA.
        {\tt\small \{zw24, chengs, gahlawat, kaacker2, lakshma2, cy45, jiahaoy5, nhovakim\}@illinois.edu}}
\thanks{P. Zhao is affiliated with the
Department of Aerospace Engineering and Mechanics, University of Alabama, AL 35487, USA.   {\tt\small pan.zhao@ua.edu}}
}
\begin{document}
\bstctlcite{BSTcontrol}
\maketitle
\thispagestyle{plain}
\pagestyle{plain}

\begin{abstract}
Quadrotors that can operate predictably in the presence of imperfect model knowledge and external disturbances are crucial in safety-critical applications. 
We present \bm{$\mathcal{L}_1$}Quad, a control architecture that ensures uniformly bounded transient response of the quadrotor's uncertain dynamics on the special Euclidean group SE(3). 
By leveraging the geometric controller and the \bm{$\mathcal{L}_1$} adaptive controller, the \bm{$\mathcal{L}_1$}Quad architecture provides a theoretically justified framework for the design and analysis of quadrotor's tracking controller in the presence of nonlinear (time- and state-dependent) uncertainties on both the translational and rotational dynamics. 
In addition, we validate the performance of the \bm{$\mathcal{L}_1$}Quad architecture through extensive experiments\footnote{Video: \href{https://youtu.be/18-2OqTRJ50}{https://youtu.be/18-2OqTRJ50}; Code: \href{https://github.com/sigma-pi/L1Quad}{https://github.com/sigma-pi/L1Quad}} for eleven types of uncertainties across various trajectories. 
The results demonstrate that the \bm{$\mathcal{L}_1$}Quad can achieve consistently small tracking errors despite the uncertainties and disturbances and significantly outperforms existing state-of-the-art controllers.
\end{abstract}

\section{INTRODUCTION}\label{sec: introduction}
In recent years, unmanned aerial vehicles have seen an increased use across a wide range of applications. The quadrotor platform, in particular, has garnered interest due to its low cost, relatively simple mechanical structure, and dynamic capabilities. Controller synthesis for quadrotors is a challenging problem due to the unstable and underactuated nature of the dynamics. The challenges are exacerbated further due to uncertainties and disturbances (e.g., wind, payload sloshing, and system degradation), potentially leading to a loss of predictability and even stability~\cite{emran2018review}.

To design controllers for quadrotors' underactuated dynamics that evolve on the special Euclidean group, SE(3), early studies adopt the linearized quadrotor dynamics at the equilibrium point around hover state~\cite{castillo2004stabilization,bouabdallah2005towards}, which can lead to poor performance. Nonlinear controllers such as backstepping~\cite{madani2006backstepping,bouabdallah2005backstepping,labbadi2019robust} and feedback-linearization~\cite{lee2009feedback,voos2009nonlinear} have also been investigated. Since these methods employ Euler angles as the attitude representation that suffers from gimbal lock~\cite{hilkert2008inertially}, they cannot perform aggressive maneuvers. Quaternions-based attitude representation has also been investigated~\cite{fresk2013full}. However, this method may cause instability to quadrotor dynamics due to unwinding phenomenon~\cite{chaturvedi2011rigid,bhat2000topological}. To address this issue, geometric control theory~\cite{jurdjevic1997geometric} has been adopted to derive controllers for quadrotors in~\cite{mellinger2011minimum,lee2010geometric}. This type of controller uses rotation matrices to represent attitudes that intrinsically characterize the geometric properties of nonlinear manifolds, thereby completely avoiding singularities and reducing design complexities. Following this approach, the authors in~\cite{lee2010arxiv} established exponential stability for the quadrotor nominal dynamics on SE(3).

Unavoidable uncertainties and disturbances in real applications further complicate quadrotor controller design. Robust and adaptive control methods have been employed to safely operate quadrotors subject to uncertainties~\cite{antonelli2017adaptive,goodarzi2015geometric,dydek2012adaptive,cabecinhas2014nonlinear}. However, these methods assume uncertainties to be linearly dependent on known basis functions, which is restrictive. Disturbance-observer-based (DOB) methods can handle a broader class of uncertainties~\cite{besnard2012quadrotor,castillo2019disturbance,guo2022safety}, yet they often assume the disturbances are generated by an exogenous system with known parameters~\cite{chen2015disturbance}, which is also limited. Additionally, the state dependence of uncertainties is generally ignored in the theoretical analysis of DOB methods~\cite{chen2015disturbance}. The control architecture proposed by the authors in~\cite{ke2023uniform} estimates uncertainties that require numerical differentiation of noisy signals, and the estimation accuracy is not guaranteed. Moreover, an incremental nonlinear dynamic inversion control has been applied in~\cite{tal2020accurate} to compensate for aerodynamic drag in the quadrotor's high-speed flights. However, this sensor-based method requires the installation of extra sensors that are uncommon to typical quadrotor hardware.

Researchers have also investigated adopting machine learning (ML) tools to address uncertainties and disturbances in quadrotor control design. Early trials tried to employ DNNs to represent quadrotor uncertain dynamics and then use the learned models to synthesize linear quadratic regulator (LQR)~\cite{bansal2016learning} or model predictive controllers (MPC)~\cite{lambert2019low}. Recent studies use ML tools to model the uncertainties and then integrate them into conventional control methods~\cite{saviolo2023active,shi2019neural,o2022neural,chee2022knode,torrente2021data,huang2023datt}.
Despite the efforts, the data collection (e.g., it is challenging to acquire data in unknown environments) and the training process form a significant overhead before the ML tools can be deployed. Additionally, when ML-based control methodologies are used, it is hard to establish theoretical guarantees on the system's performance without placing conservative assumptions. For example, the authors in~\cite{o2022neural,shi2021neural,joshi2019deep} use the universal approximation theorem~\cite[Theorem 2.1]{hornik1989multilayer} to ensure the learned model can estimate a continuous function arbitrarily accurately. However, this theorem requires the states to remain in a compact set, which essentially assumes the quadrotor system is stable \textit{a priori}, and it is very restrictive. Furthermore, ML-based methods usually require extensive computational power and cannot be easily executed on standard quadrotor flight controllers.

In this paper, we present $\mathcal{L}_1$Quad, a quadrotor control architecture that uses the~\ellone adaptive control~\cite{hovakimyan2010L1} as an augmentation to compensate for the uncertainties. The use of~\ellone adaptive control is motivated by its architectural advantage of decoupling the estimation loop from the control loop, which allows the use of arbitrarily fast adaptation rates without sacrificing the robustness of the closed-loop system~\cite{hovakimyan2010L1}. The $\mathcal{L}_1$ adaptive control has been successfully validated on NASA's AirStar~\cite{gregory2010flight}, Calspan's Learjet~\cite{ackerman2016l1,ackerman2017evaluation}, and unmanned aerial vehicles~\cite{kaminer2010path,jafarnejadsani2017optimized}. The $\mathcal{L}_1$ adaptive controller design for quadrotors has been studied in~\cite{zuo2014augmented,huynh20141}, where Euler angles were adopted. To avoid the singularities in Euler-angles attitude representation, the authors of~\cite{kotaru2020geometric} have proposed an $\mathcal{L}_1$ adaptive control augmentation of the geometric controller that applies to the rotational dynamics only, which cannot compensate for any uncertainties in the translational dynamics.

$\mathcal{L}_1$Quad adopts the geometric controller~\cite{lee2010geometric} as the baseline controller and augments the \ellone adaptive controller~\cite{hovakimyan2010L1} to address the uncertainties and disturbances on both rotational and translational dynamics. $\mathcal{L}_1$Quad can handle uncertainties and disturbances that are nonlinearly dependent on both time and states, allowing the proposed architecture to compensate for a broader class of uncertainties. Additionally, in the proposed architecture, the estimation is decoupled from control, thereby allowing for arbitrarily fast adaptation without hurting the system's robustness. We show that with $\mathcal{L}_1$Quad, the estimation error is bounded and can be tuned arbitrarily small. Furthermore, the proposed architecture guarantees transient performance in terms of uniform bounds on the error between actual states and those of a nominal system. These uniform bounds characterize tubes centered around the desired states so that actual states are guaranteed to stay inside despite uncertainties. Moreover, we show that the size of these uniform bounds can be reduced by tuning the low-pass filter's bandwidth and the proposed architecture's sampling time.

We implement $\mathcal{L}_1$Quad on a quadrotor equipped with a Pixhawk flight controller running customized Ardupilot firmware. To demonstrate the performance of $\mathcal{L}_1$Quad, we test the proposed controller in experiments with numerous trajectories against various uncertainties, including injected uncertainty, slosh payload, chipped propeller, mixed propellers, ground effect, voltage drop, downwash, tunnel, and hanging off-center weights. 
$\mathcal{L}_1$Quad demonstrates consistently smaller tracking errors than several other  controllers in comparison~\cite{lee2010geometric,goodarzi2013geometric,goodarzi2015geometric}. We also conduct benchmark experiments with uncertainties of gradually changing magnitudes at different flying speeds to demonstrate the consistently superior performance of the $\mathcal{L}_1$Quad in dynamic environments. It is worth noting that \textit{only} one set of parameters of $\mathcal{L}_1$Quad is used for all the experiments above: {\em no retuning is needed.}

We summarize our contributions from two perspectives: i) from the perspective of $\mathcal{L}_1$ adaptive control of quadrotors~\cite{jafarnejadsani2017optimized,zuo2014augmented,huynh20141,kotaru2020geometric}, this is the first paper that considers nonlinear reference systems on the SE(3) manifold, thus allowing to capture -- in the problem formulation -- more challenging trajectories for tracking and hence the broader class of uncertainties. The earlier papers considered only linearized models or rotational dynamics on SO(3), and thus, they had limited operational envelopes and dealt with a limited class of uncertainties. 
The extension to nonlinear reference systems, from the perspective of $\mathcal{L}_1$ adaptive control theory, was analyzed in~\cite{wang2012l1,lakshmanan2020safe} with the projection operator as the adaptation law. However, since the projection operator-based adaptation law suffers from numerical instability~\cite{campbell2010implementation}, the current paper employs the piecewise-constant adaptation law to alleviate this issue. The modification was not trivial and required the development of Proposition~\ref{prop: estimation error bound} in the current paper.
ii) From the perspective of quadrotor control design under uncertainties, we observe that the existing papers~\cite{o2022neural,bisheban2020geometric,tal2020accurate,torrente2021data,shi2019neural,wang2023neural,saviolo2023active} have not demonstrated the capability of handling multiple classes of uncertainties \textbf{simultaneously by a single design};  neither they have developed stability and performance guarantees. These classes include external disturbances, actuation-induced uncertainties, and model mismatch, which cover the majority of uncertainties and disturbances experienced by a quadrotor.
Notice that in the current paper, \textbf{with a single set of control parameters}, $\mathcal{L}_1$Quad consistently demonstrates favorable flight performance under eleven different types of uncertainties across various agile trajectories without retuning and without modeling or enforcing parametric structure on the uncertainties. Moreover, we conduct theoretical analysis directly on the SE(3) manifold and provide tunable uniform performance bounds, which guarantee both the transient and steady-state performance.

This article significantly extends the work in~\cite{wu20221}, which presented preliminary results on a Parrot Mambo quadrotor. Key additions include the performance analysis in the second half of Section~\ref{sec: geometric controller} and the entirety of Section~\ref{sec: performance analysis}. Moreover, we provide a more elaborate description and implementation procedure of this architecture, as well as new experimental results. These experiments are conducted on a custom-built quadrotor platform with a Pixhawk flight controller, enabling validation of the proposed controller under a broader spectrum of uncertainties and disturbances when tracking more agile trajectories. This provides a more comprehensive evaluation of the proposed control method.
 
The remainder of the paper is organized as follows: Section~\ref{sec: geometric controller} introduces the background on the quadrotor dynamics, as well as the design and analysis of the geometric controller. Section~\ref{sec: uncertainty characterization} explains the modeling of uncertainties. Section~\ref{sec:geometric+L1} shows the $\mathcal{L}_1$ adaptive augmentation of the geometric controller. Section~\ref{sec: performance analysis} provides performance guarantees of the proposed control framework. Section~\ref{sec:experiments} shows the experimental results conducted on a custom-built quadrotor. Finally, Section~\ref{sec:conclusions} summarizes the paper and discusses the limitations of the proposed method and future work.  

The following notations are used throughout the paper:
We denote the vectorization of a matrix $A$ by $\vect(A)$, which is obtained by stacking the columns of the matrix $A$ on top of one another. We use $\mathbb{N}$ and $\mathbb{R}$ to represent natural and real numbers, respectively. The \textit{wedge} operator $\cdot^{\wedge}:\mathbb{R}^3 \rightarrow \mathfrak{so}(3)$ denotes the mapping to the space of skew-symmetric matrices. The \emph{vee} operator $\cdot^\vee$ is the inverse of the \textit{wedge} operator which maps $\mathfrak{so}(3)$ to $\mathbb{R}^3$. We use $\text{tr}(A)$ to denote the trace of a matrix $A$. We use $(v)_i$ to represent the $i^{th}$ element of a vector $v$. The largest and smallest eigenvalue for a square matrix $A$ are denoted by $\lambda_M(A)$ and $\lambda_m(A)$, respectively. The notation $A \succ 0$ means a square matrix $A$ is positive definite. We use $\diag{A,\dots,B}$ to denote the block diagonal matrix comprised of matrices $A,\dots,B$. 

\section{Quadrotor Dynamics and Geometric controller}\label{sec: geometric controller}
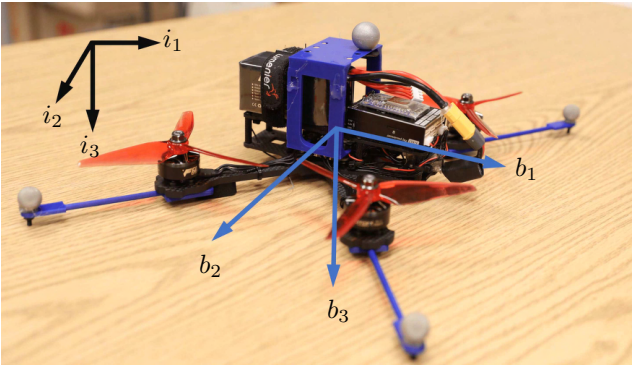
\begin{figure}
\centering
\begin{tikzpicture}
    \node (img) {\includegraphics[width=0.97\columnwidth]{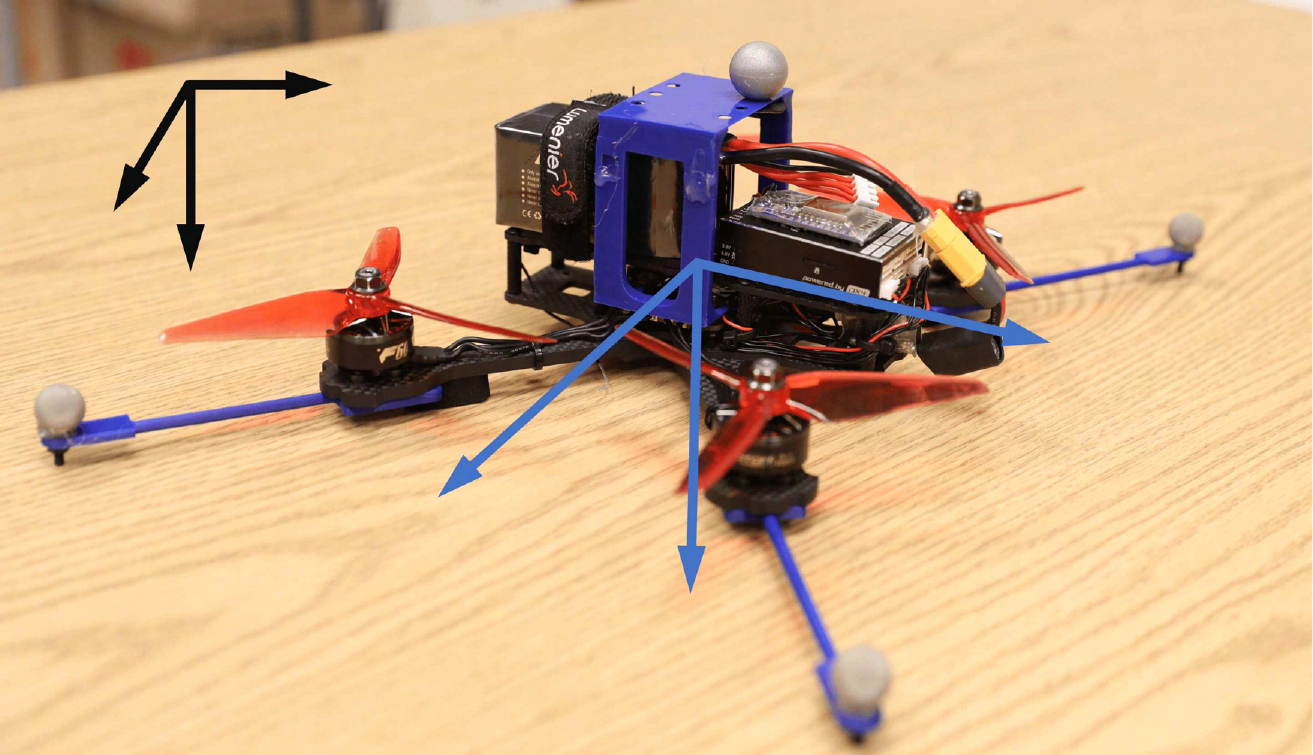}};
    \coordinate (b1) at (2.8,0.2);
    \coordinate (b2) at (-1.4,-1.1);
    \coordinate (b3) at (0.3,-1.7);
    \node[] at (b1) {$b_1$};
    \node[] at (b2) {$b_2$};
    \node[] at (b3) {$b_3$};
    \coordinate (i1) at (-1.9,1.9);
    \coordinate (i2) at (-3.5,0.9);
    \coordinate (i3) at (-3.0,0.5);
    \node[] at (i1) {$i_1$};
    \node[] at (i2) {$i_2$};
    \node[] at (i3) {$i_3$};
\end{tikzpicture}
    \caption{The quadrotor and reference frames.}
    \label{fig:quadrotor with reference frames}\vspace{-5mm}
\end{figure}

We choose an inertial frame and a body-fixed frame, which are spanned by unit vectors $\{ i_1,i_2,i_3 \}$ and $\{ b_1,b_2,b_3 \}$ in north-east-down directions, respectively, shown in Fig~\ref{fig:quadrotor with reference frames}. The origin of the body-fixed frame is located at the center of mass (COM) of the quadrotor, which is assumed to be its geometric center.
The configuration space of the quadrotor is defined by the location of its COM and the attitude with respect to the inertial frame, i.e., the configuration manifold is the special Euclidean group $SE(3)$.
By definition, $b_i = Re_i$, where $R \in SO(3)$ is a rotation matrix, and $b_i, \, e_i \in \mathbb{R}^3$, $i \in \{1,2,3\}$, are the $i^{th}$ body-fixed axis and the $i^{th}$ canonical basis vector, respectively. 
The dynamics of quadrotor in the absence of uncertainties are given by
\begin{subequations}\label{equ:quadrotor dynamics}
\begin{align}
    \dot{p}= & v,\label{eq:position integral} 
    \\
    \dot{v}= & ge_3-\frac{f_b}{m}Re_3,\label{equ:translational dynamics}
    \\
    \dot{R} = & R \Omega^{\wedge},\label{eq:angle integral}
    \\
     \dot{\Omega}= & J^{-1}(M_b - \Omega \times J\Omega), \label{equ:rotational dynamics}
\end{align}
\end{subequations}
where $p \in \mathbb{R}^3$ and $v \in \mathbb{R}^3$ are the position and velocity of the quadrotor's COM in the inertial frame, respectively, ${\Omega \in \mathbb{R}^3}$ is the angular velocity in the body-fixed frame, $g$ is the gravitational acceleration, $m$ is the vehicle mass, $J \in \mathbb{R}^{3 \times 3}$ is the moment of inertia matrix calculated in the body-fixed frame, $f_b$ is the collective thrust, and $M_b \in \mathbb{R}^3$ is the moment in the body-fixed frame. Note that we suppress temporal dependencies to maintain clarity of exposition unless required.

We do not consider the motor dynamics and the propellers' aerodynamic effects. Therefore, the thrust $f_b$ and moment $M_b$ are assumed to be linear in the squared motor speeds~\cite{mellinger2011minimum}. We choose $f_b$ and $M_b$ as the control in~\eqref{equ:quadrotor dynamics}, and they are achieved via motor mixing (see \cite{mellinger2011minimum} for details in computation).
The design of the geometric controller follows \cite{lee2010geometric,mellinger2011minimum}, where the goal is to have the quadrotor follow prescribed trajectory $p_d(t) \in \mathbb{R}^3$ and yaw $\psi_d(t)$ for time $t$ in a prescribed interval $[0,t_f]$. Towards this end, the translational motion is controlled by the desired thrust $f_b$:
\begin{equation}\label{eq:thrust control}
    f_b = -F_d \cdot (Re_3),
\end{equation}
which is obtained by projecting the desired force $F_d$ to the body-fixed $z$-axis $Re_3$. The desired force $F_d \in \mathbb{R}^3$ is computed via $F_d = -K_p e_p-K_v e_v - mg e_3 + m \ddot{p}_d$ with $K_p,K_v \in \mathbb{R}^{3 \times 3}$ being user-selected positive-definite gain matrices and $e_p = p-p_d$ and $e_v = \dot{p}-\dot{p}_d$ standing for the position and velocity errors, respectively.
The rotational motion is controlled by the desired moment
\begin{align}\label{eq:torque control}
    M_b =& -K_R e_R - K_\Omega e_\Omega + \Omega \times J \Omega  \nonumber\\
     &- J( \Omega^{\wedge} R^\top R_d \Omega_d-R^\top R_d \dot{\Omega}_d),
\end{align}
where $K_R,$ $K_{\Omega} \in \mathbb{R}^{3 \times 3}$ are user-selected positive-definite gain matrices; $R_d$, $\Omega_d$, and $\dot{\Omega}_d$ are the desired rotation matrix, desired angular velocity, and desired angular velocity derivative, respectively; $e_R = ( R_d^\top R - R^\top R_d)^{\vee}/2$ and $e_{\Omega} = \Omega -R^\top R_d \Omega_d$ are the rotation error and angular velocity error, respectively. The derivations of $R_d$, $\Omega_d$, and $\dot{\Omega}_d$ are omitted for simplicity; see~\cite{lee2010arxiv} for details. 

Before we revisit the performance analysis of the geometric controller, we first define the state $x \in \mathbb{R}^{3} \times \mathbb{R}^3 \times SO(3) \times \mathbb{R}^3 $ by $x(t) = \left(p(t),v(t),R(t),\Omega(t)\right)$, and the desired state $x_d(t) = (p_d(t),v_d(t),R_d(t),\Omega_d(t))$ which satisfies the nominal dynamics in~\eqref{equ:quadrotor dynamics}.  Further, let
\begin{align}\label{equ: norm of tuples}
    &d(x(t),x_d(t)) \nonumber \\ &\triangleq \norm{\begin{bmatrix}
e_p(t)^\top& e_v(t)^\top &e_R(t)^\top & e_\Omega(t)^\top
\end{bmatrix}^\top},
\end{align}
which describes the distance between the true states and the desired states.
\begin{proposition}\label{prop:geometricproposition}
Consider the thrust magnitude $f_b$ and moment vector $M_b$ defined by equations~\eqref{eq:thrust control} and~\eqref{eq:torque control}, respectively. Suppose the initial condition $x(0) = \left(p(0),v(0),R(0),\Omega(0)\right)$ satisfies 
\begin{equation}\label{equ: condition on psi0}
    \Psi(R(0),R_d(0)) < 1,
\end{equation}
for $\Psi(R,R_d)\triangleq \text{tr}(I-R_d^\top R)/2$ being the rotation error function.
Define $W_1$, $W_{12}$, $W_2 \in \mathbb{R}^{2 \times 2}$ as follows:
\begin{equation*}\label{equ: define W1}
W_1 = \left[ \begin{smallmatrix} \frac{c_1 K_p}{m}(1-\alpha) &  -\frac{c_1 K_v}{2m}(1+\alpha)-\frac{K_p \alpha}{2}\\  -\frac{c_1 K_v}{2m}(1+\alpha)-\frac{K_p \alpha}{2} & (K_v(1-\alpha)-c_1) \end{smallmatrix} \right],
\end{equation*}
\[
W_{12} = \left[ \begin{smallmatrix}  \frac{c_1}{m}H & 0\\ H &  0\end{smallmatrix} \right], \quad W_2 = \left[ \begin{smallmatrix}  \frac{c_2 K_R}{\lambda_M(J)} & -\frac{c_2 K_\Omega}{2 \lambda_m(J)} \\ -\frac{c_2 K_\Omega}{2 \lambda_m(J)} & K_\Omega -c_2 \end{smallmatrix} \right],
\]
and define $M_{11},M_{12},M_{21},M_{22} \in \mathbb{R}^{2 \times 2}$ as
\begin{align*}
    M_{11}=&\frac{1}{2} \begin{bmatrix} Kp & -c_1 \\ -c_1 & m \end{bmatrix}, & 
    M_{12}=&\frac{1}{2} \begin{bmatrix} Kp & c_1 \\ c_1 & m \end{bmatrix}, \\
    M_{21}=&\frac{1}{2} \begin{bmatrix} K_R & -c_2 \\-c_2 & \lambda_m(J) \end{bmatrix}, & M_{22} =& \frac{1}{2} \begin{bmatrix}
        \frac{2K_R}{2-\psi_1} & c_2 \\c_2 & \lambda_M(J)
    \end{bmatrix},
\end{align*}
where $\psi_1$ is such that $ \Psi \left(R(0),R_d(0)\right) < \psi_1 < 1$, $H > \norm{-mge_3 + m \Ddot{x}_d}$, and $\alpha = \sqrt{\psi_1(2-\psi_1)}$. 
If we choose positive constants $K_p$, $K_v$, $K_{\Omega}$, $K_R$, $c_1$, $c_2$ such that $M_{11}$, $M_{21}$, $M_{12}$, $M_{22}$, $W_1$, $W_2$ are positive definite matrices, and
\begin{equation}\label{equ:definition of W}
W = 
\begin{bmatrix}
W_1 & -\frac{W_{12}}{2} \\ -\frac{W_{12}^\top}{2} & W_2
\end{bmatrix} \succ 0,
\end{equation}
then the equilibrium  $\left( e_p, e_v, e_R, e_\Omega \right) = \left( 0,0,0,0  \right)$ is exponentially stable for any initial condition satisfying~\eqref{equ: condition on psi0}, and
\begin{equation}\label{equ: region of attraction}
    \norm{e_\Omega(0)}^2 < \frac{2}{\lambda_M(J)}K_R(\psi_1 - \Psi(R(0),R_d(0))).
\end{equation}
Particularly, the positive definite function
\begin{align}\label{equ:lyapunovforcomplete}
&V(e_{p}(t),e_{v}(t),e_{R}(t),e_{\Omega}(t)) \nonumber
\\
& \begin{aligned}[t]= &\frac{1}{2}K_p \norm{e_p(t)}^2 +   \frac{1}{2}e_\Omega(t) \cdot Je_\Omega(t) + \frac{1}{2}m \norm{e_v(t)}^2 \nonumber 
\\  
&+ c_1 e_p(t)\cdot e_v(t) + c_2 e_R(t) \cdot e_\Omega(t)\end{aligned}
\\ 
&\hspace{2.8mm}+ K_R \Psi (R(t),R_d(t)) 
\end{align}
is the Lyapunov function associated with the error dynamics, as it satisfies
\begin{equation}\label{equ:lyapunovdotsmallerthanlyapunov}
    \dot{V} \leq -\beta V,
\end{equation}
and
\begin{equation}\label{equ: lyapunovboundfinal}
    \underline{\gamma}d(x,x_d)^2 \leq V \leq \overline{\gamma}d(x,x_d)^2,
\end{equation}
where
\begin{subequations}
\begin{align} 
    \beta \leq& \frac{\lambda_m(W)}{\lambda_M \left( \diag{M_{12} , M_{22}} \right)},\label{equ: beta condition} \\
    \underline{\gamma} \triangleq& \min(\lambda_m(M_{11}),\lambda_m(M_{21})),\label{equ:definition of gamma lower bar}\\
\overline{\gamma} \triangleq &\max\left(\lambda_M(M_{12}\right),\lambda_M\left(M_{22})\right)\label{equ:definition of gamma upper bar}.
\end{align}
\end{subequations}
\end{proposition}
The details of the error dynamics and the proof of this proposition can be found in Appendix~\ref{apdx: proof of geometric}. Note that the proof of Proposition~\ref{prop:geometricproposition} relaxes a simplifying assumption in~\cite[Appendix D, Section (c)]{lee2010arxiv}.

\section{Uncertainty Characterization}\label{sec: uncertainty characterization}
The nominal dynamics \eqref{equ:quadrotor dynamics} provide a description of the quadrotor's motion in an ideal case. However, in reality, a quadrotor's performance is affected by uncertainties and disturbances, such as propellers' aerodynamic effects, ground effects, etc. Our approach lumps the uncertainties induced from various sources into uncertain forces and moments. To see how it works, we first introduce the state-space representation of the nominal dynamics in~\eqref{equ:quadrotor dynamics}:
\begin{equation}\label{equ:nominal system with partial state}
    \dot{x}(t)=f_F(x(t)) +B_F(R(t))u_b(t), \quad x(0)=x_0,
\end{equation}
where
\begin{equation*}
\vspace{-1mm}
    f_F(x)= \begin{bmatrix}
          v \\
        ge_3 \\
         \vect(R\Omega^{\wedge}) \\
        -J^{-1}\Omega \times J \Omega
    \end{bmatrix}, \ B_F(R)=  \begin{bmatrix}
        0_{3\times1} & 0_{3\times3}\\
        -m^{-1}Re_3 & 0_{3\times3}\\
        0_{9\times1} & 0_{9\times3}\\
        0_{3\times1} & J^{-1}
    \end{bmatrix},
\end{equation*}
$u_b = [f_b \ M_b^\top]^\top$ is the control input from the baseline geometric controller and $x_0 \in \mathbb{R}^{18}$ is the initial state vector. 
Note that there is an abuse of notation that if $x$ is used in state-space representations, e.g., equation~\eqref{equ:nominal system with partial state}, it takes the vector form $x=[p^\top \ v^\top \ \vect(R)^\top \ \Omega^\top ]^\top$. Otherwise, it takes the from $x=(p,v,R,\Omega)$.

We only consider the model uncertainties and external disturbances. To be specific, the uncertainties enter the system \eqref{equ:quadrotor dynamics} through the dynamics via \eqref{equ:translational dynamics} and \eqref{equ:rotational dynamics}. Therefore, the kinematics \eqref{eq:position integral} and \eqref{eq:angle integral} are integrators and thus uncertainty-free. As a result, we have the uncertain dynamics:
\begin{align}\label{equ:uncertain system with full state}
    \dot{x}(t)=&f_F(x(t)) + B_F(R(t))(u_b(t) + \sigma_m(t,x(t))) \nonumber \\ &+ B_F^\bot(R(t)) \sigma_{um}(t,x(t)), \quad x(0)=x_0,
\end{align}
where $\sigma_m \in \mathbb{R}^4$, $\sigma_{um} \in \mathbb{R}^2$ stand for the matched and unmatched uncertainties, respectively, and 
\[
B_F^{\bot}(R)=  \begin{bmatrix}
        0_{3\times1} & 0_{3\times1} \\
        m^{-1}Re_1 &m^{-1}Re_2\\
        0_{12\times1} & 0_{12\times1}
    \end{bmatrix}.
\]
The matched uncertainty $\sigma_m$ enters the system in the same channel as the baseline control input $u_b$, which contains the forces along the body-$z$ axis and moments in body-$x$, -$y$, and -$z$ axes. Therefore, $\sigma_m$ can be directly compensated for by the control actions. In contrast, the unmatched uncertainty $\sigma_{um}$ contains forces along any direction in the  body-$xy$ plane, which enter the dynamics through $B^{\bot}_F(R)$ (whose columns are perpendicular to those of $B_F(R)$). 

Our formulation with $\sigma_m$ and $\sigma_{um}$ in~\eqref{equ:uncertain system with full state} can characterize various types of uncertainties and disturbances which can be categorized into three classes: i) external disturbances, ii) actuation-induced uncertainties, and iii) model mismatch.
External disturbances (e.g., wind, gust, and ground effect) apply to the quadrotor's dynamics in the form of unknown force $F_0 \in \mathbb{R}^3$ and moment $M_0 \in \mathbb{R}^3$. The matched uncertainty $\sigma_m$ contains the unknown force $F_0$'s projection onto the body-$z$ axis and the unknown moment $M_0$ (i.e., $\sigma_m = [F_0^\top Re_3 \ M_0^\top]^\top$), whereas the unmatched uncertainty~$\sigma_{um}$ contains $F_0$'s projection onto the body-$xy$ plane (i.e., $\sigma_{um} = [F_0^\top Re_1\ F_0^\top R e_2]^\top$). The actuation-induced uncertainties (e.g., battery voltage drop and damaged propellers) will cause scaled actual thrust $f_{real}$ and moment $M_{real}$ to deviate from the commanded thrust $f_{b}$ and moment $M_{b}$. The deviation can be captured by $\sigma_m = [f_{real} \  M_{real}^\top]^\top - [f_b \ M_b^\top]^\top$. For the model mismatch, we provide a simple example where the mass of the quadrotor in design is $m$ while the true value of the mass is $m_{real}$. The effect of this mismatch on the quadrotor's translational dynamics is reflected by
\[
\dot{v}=  ge_3-\frac{f_b}{m_{real}}Re_3 =  ge_3-\frac{f_b}{m}Re_3-\frac{(m-m_{real})f_b}{m_{real}m} Re_3.
\]
Comparing the above equation with the ideal form in~\eqref{equ:translational dynamics}, we can use $\sigma_m$ in~\eqref{equ:uncertain system with full state} to describe this effect where $(\sigma_m)_1 = (m-m_{real})f_b / m_{real}$. 

Notice that both types of uncertainties, $\sigma_m(t,x(t))$ and $\sigma_{um}(t,x(t))$, are modeled as non-parametric uncertainties, i.e., they are not necessarily parameterized with a finite number of parameters~\cite{boyd1986note} (in the form of $\xi^\top W(x)$ with $\xi \in \mathbb{R}^r$ being unknown parameters and $W(x) \in \mathbb{R}^r$ being known vector-valued nonlinear functions of $x$ forming a basis). Another important feature is that both uncertainties are nonlinearly dependent on time and state, which is the case for a broad class of uncertainties in practice.\vspace{-2mm}
\section{$\mathcal{L}_1$ Adaptive Augmentation}\label{sec:geometric+L1}
\vspace{-1mm}
In this section, we describe the $\mathcal{L}_1$ adaptive augmentation for the quadrotor control system. The $\mathcal{L}_1$ adaptive control input $u_{ad} = [f_{ad} \  M_{ad}^\top ]^\top \in \mathbb{R}^4$ enters the system in the same channel as the baseline control input $u_b$. As a result, we have
\begin{align}\label{equ:uncertain system with full state and with L1}
    \dot{x}(t)=&f_F(x(t)) + B_F(R(t))(u_b(t) + u_{ad}(t) + \sigma_m(t,x(t))) \nonumber \\ &+ B_F^\bot(R(t)) \sigma_{um}(t,x(t)), \quad x(0)=x_0.
\end{align}
The objective is to ensure that the state $x(t)$ of the system in~\eqref{equ:uncertain system with full state and with L1} remains `close' to the desired trajectory $x_d(t)$, for all $t \geq 0$.

Next we introduce the notion of `closeness'. Given a positive scalar $\rho$ and the desired state $x_d(t)$, $\mathcal{O}(x_d(t),\rho)$ denotes the set centered at $x_d(t)$ with radius $\rho$, i.e., $\mathcal{O}(x_d(t),\rho)=\{ x(t) \in \mathbb{R}^3 \times \mathbb{R}^3 \times SO(3) \times \mathbb{R}^3 | 
d(x(t),x_d(t)) \leq \rho \}$. Clearly $\mathcal{O}(x_d(t),\rho)$ induces a tube centered around $x_d(t)$ for all $t>0$, where the tube is defined as
\begin{equation}\label{equ: tube definition}
    \mathcal{T}(\rho) \triangleq \cup_{t\geq 0} \mathcal{O}(x_d(t),\rho).
\end{equation}
Thus, $u_{ad}(t)$ needs to ensure
\[
x(t) \in \mathcal{O}(x_d(t),\rho), \quad \forall t \geq 0,
\]
for any given $x_d(t)$ and  $\rho>0$.
\begin{assumption}\label{aspt: main assumption}
$\sigma(t,x(t))$ is continuous, bounded, and Lipschitz continuous in $x$, uniformly in $t$, for all $x \in \mathcal{T}(\rho)$ and $t \geq 0$. In addition, $\frac{\partial \sigma}{\partial x}(t,x(t))$ and $\frac{\partial \sigma}{\partial t}(t,x(t))$ are bounded for all $x \in \mathcal{T}(\rho)$ and $t \geq 0$.
\end{assumption}
Based on Assumption~\ref{aspt: main assumption} and considering the fact that $\sigma = [{\sigma}_m^\top \ {\sigma}_{um}^\top]^\top$, we have that for all $x \in \mathcal{T}(\rho)$ and all $t \geq 0$:
\begin{equation*}
\noindent\centering
\begin{minipage}{0.22\textwidth}
\begin{subequations}
\begin{align}
\norm{\frac{\partial \sigma(t,x)}{\partial t}}&\leq L_{\sigma_{t}}, \\ \setcounter{equation}{2} \norm{\frac{\partial \sigma(t,x)}{\partial x}}&\leq L_{\sigma_{x}},
\end{align}
\end{subequations}
\end{minipage}
\begin{minipage}{0.25\textwidth}
\setcounter{equation}{15}
\begin{subequations}
\setcounter{equation}{1}
\begin{align}
\norm{\frac{\partial \sigma_m(t,x)}{\partial t}}&\leq L_{\sigma_{m_t}},\\\setcounter{equation}{3}\norm{\frac{\partial \sigma_m(t,x)}{\partial x}}&\leq L_{\sigma_{m_x}},
\end{align} \nonumber
\end{subequations}
\end{minipage}\vspace{2mm}
\end{equation*}

\begin{equation*}
\noindent\centering
\begin{minipage}{0.24\textwidth}
\setcounter{equation}{15}
\begin{subequations}
\begin{align}
\setcounter{equation}{4}
\norm{\sigma(t,x)} \leq& \Delta_{\sigma},\label{equ: bounds of sigmas}\\ \setcounter{equation}{6}\norm{\sigma_{um}(t,x)} \leq& \Delta_{\sigma_{um}}.
\end{align}
\end{subequations}
\end{minipage}
\begin{minipage}{0.24\textwidth}
\vspace{-5mm}
\setcounter{equation}{15}
\begin{subequations}
\setcounter{equation}{5}
\begin{align}
\norm{\sigma_{m}(t,x)} \leq& \Delta_{\sigma_{m}},
\end{align} \nonumber
\end{subequations}
\end{minipage}\vspace{2mm}
\end{equation*}
\vspace{-1mm}
\begin{assumption}\label{asump: unmatched uncertainty assumption}
The constant $\Delta_{\sigma_{um}}$ verifies
\begin{align}
\Delta_{\sigma_{um}} <& \frac{\underline{\gamma}\rho^2-V_0}{c_4\rho},\label{unmatched special bound}
\end{align}
where $c_4 \triangleq \max\{c_1/m,1\}$, $c_1$, and $\underline{\gamma}$ are defined in Proposition~\ref{prop:geometricproposition}, and $V_0 = V(e_{p}(0),e_{v}(0),e_{R} (0),e_{\Omega}(0))$.
\end{assumption}
\begin{remark}
If the Lipschitz constant of $\sigma_{um}(t,x(t))$ with respect to the state $x$ is small enough, the upper bound in~\eqref{unmatched special bound} in Assumption~\ref{asump: unmatched uncertainty assumption} can always be satisfied by choosing a large enough $\rho$. 
In real application, side wind is the most common cause of unmatched uncertainty for quadrotors. 
It is realistic to assume the magnitude of the side wind does not change dramatically within a relatively large region. Note that choosing a sufficiently large $\rho$ to satisfy the condition in~\eqref{unmatched special bound} may lead to a degradation of the performance bound.
\end{remark}

\begin{remark}
The implementation of the $\mathcal L_1$Quad controller does not require knowledge of the bounds in Assumptions~\ref{aspt: main assumption} or~\ref{asump: unmatched uncertainty assumption}, as illustrated in Algorithm~\ref{alg: L1 algorithm}. The knowledge of these bounds is needed only for analysis.
\end{remark}

Since the kinematics \eqref{eq:position integral} and \eqref{eq:angle integral} are uncertainty-free, we design the $\mathcal{L}_1$ adaptive controller only using the dynamics~\eqref{equ:translational dynamics} and~\eqref{equ:rotational dynamics} that are directly affected by the uncertainties. The state-space representation of \eqref{equ:translational dynamics} and~\eqref{equ:rotational dynamics} is given by:
\begin{align}\label{equ:quadrotor partial states uncertain dynamics}
    \dot{z}(t)=&f(z(t))+B(R(t))\left(u_b(t) + u_{ad}(t)+\sigma_m(t,x(t))\right) \nonumber \\ &+B^{\bot}(R(t))\sigma_{um}(t,x(t)), \quad z(0)=z_0,
\end{align}
where $z = [v^\top \ \Omega^\top]^\top \in \mathbb{R}^6$, $f(z)= \left[\begin{smallmatrix}
        ge_3 \\
        -J^{-1}\Omega \times J \Omega
    \end{smallmatrix}\right]$, $B(R)=\left[  \begin{smallmatrix}
        -m^{-1}Re_3 & 0_{3\times3}\\
        0_{3\times1} & J^{-1}
    \end{smallmatrix} \right]$, $B^{\bot}(R)= \left[ \begin{smallmatrix}
        m^{-1}Re_1 &m^{-1}Re_2\\
        0_{3\times1} & 0_{3\times1}
    \end{smallmatrix}\right]$, and $z_0 \in \mathbb{R}^6$ is the initial state vector.

\begin{figure}[t]
\vspace{2mm}
\begin{center}
    \begin{tikzpicture}[scale=0.78, transform shape,align=center] 
     \filldraw[fill=lcolor!70, densely dotted] (1.1,-0.9) -- (8,-0.9) -- (8,-4) -- (1.1,-4)  -- cycle;
    
       \node [input, name=input] {};
        \node [block] (ccm)  {Geometric\\Control};
        \node [sum] (sum) [right=0.8cm of ccm] {};
        \node [block] (sys) [right=1.5cm of sum] {Uncertain \\ Quadrotor System};
        \node [sqblock] (filt) [below=1.4cm of sum] {LPF};
          \node [block] (planner) [below=1.1cm of ccm] {Planner};
        \node [block] (pred) [below=1.4cm of sys] {State Predictor};
        \node [block] (adap) [below=0.5cm of pred] {Adaptation Law};
        \node [sum] (diff) [right=1.0cm of pred] {};
        \node [sum] (sum2) [left=0.7cm of pred] {};
        \node [phantom] (p2) at (sys -| diff) {};
        \node [phantom, above=1.0cm of p2] (p3) {};
        \node [phantom, below=1.9cm of p2] (p4) {};
        \node [phantom, above=1cm of p2] (p11) {};
        \node [phantom, right=1.3cm of sum] (p6) {};
        \node [phantom, below=1.9cm of p6] (p7) {};
        \node [phantom, above=2.5cm of sys] (p8) {};
        \node [phantom] (p9) at (adap -| sum2) {};

        \draw[-Latex] (planner) -- node [midway,left] {$(p_d,\psi_d)$} (ccm);
        \draw[-Latex] (ccm) -- node [midway,above] {$u_{b}$} (sum);
        \draw[-Latex] (sum) -- node [midway,above] {$u$} (sys);
         \draw[-Latex] (sum) -| (sum2);
        \draw[-Latex] (filt) -- node [near start,right] {$u_{ad}$} node[pos=0.99,right] {$+$} (sum);
        \draw[-Latex] (p9) -|  node[pos=0.99,right] {$+$} node [near end,right] {$\hat{\sigma}$}(sum2);
        \draw[-Latex] (pred) -- node [midway,above] {$\hat{x},\hat{z}$} (diff);
        \draw[-Latex] (diff) |- node [near start,right] {$\tilde{x},\tilde{z}$} (adap);
        \draw[-Latex] (p2) -- node[pos=0.99,right] {$-$} (diff);
        \draw (sys) -- node [midway,above] {$x,z$} (p2);
        \node [output, right=1.0cm of p2] (output) {};
        \draw[-Latex] (p2) -- (output);
        \draw[-Latex] (sum2) -- (pred);
        \draw (p2) -- (p11);
        \draw[-Latex] (p11) -| (ccm);
        \draw (adap) -- (p9);
        \draw[-Latex] (p9) -| node [near end,right] {$\hat{\sigma}_m$}(filt);
    \end{tikzpicture}
    \vspace{-2mm}
  \caption{The framework of geometric control with $\mathcal{L}_1$ adaptive augmentation for quadrotors. The $\mathcal{L}_1$ adaptive controller is highlighted in blue. }
  \label{fig: L1 control framework}
 \end{center}
 \vspace{-5mm}
\end{figure}
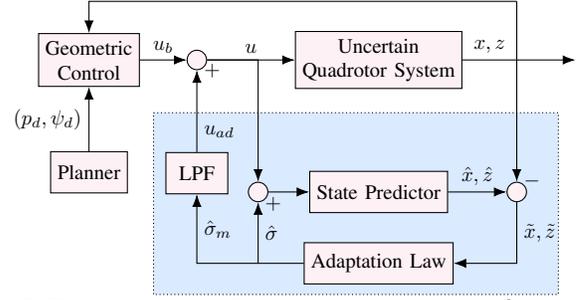

The $\mathcal{L}_1$ adaptive controller includes a state predictor, an adaptation law, and a low-pass filter (LPF) as shown in Fig~\ref{fig: L1 control framework}. 
The state predictor is given by
\begin{align}\label{equ:L1statepredictor}
    \dot{\hat{z}}(t)=&f(z(t)) +B(R(t))(u_b(t) +  u_{ad}(t)) \nonumber \\ &+h(t)+A_s\Tilde{z}(t), \quad \hat{z}(0)=z_0,
\end{align}
where $\Tilde{z}=\hat{z}-z$ is the prediction error, $A_s \in \mathbb{R}^{6 \times 6}$ is a user-selected diagonal Hurwitz matrix, $h(t) \in \mathbb{R}^6$ represents the estimation of $\bar{B}(R(t))\sigma(t,x(t))$, and $\bar{B}(R(t))=[B(R(t)) \ B^\bot(R(t))]$. 

We use the piecewise-constant adaptation law to compute $h(t)$ such that for $t \in [iT_s,(i+1)T_s)$, and $i \in \mathbb{N}$, we have
\begin{equation}\label{equ:L1adaptationlaw}
h(t)
=
h(iT_s)
=-\Phi^{-1}\mu(iT_s),\vspace{-1mm}
\end{equation}
where $T_s$ is the sampling time, $\Phi=A_s^{-1}(\exp(A_sT_s)-I)$, and $\mu(iT_s)=\exp(A_sT_s)\Tilde{z}(iT_s)$. With the piecewise-constant adaptation law, it is straightforward to calculate the estimated uncertainty for $t \in [iT_s,(i+1)T_s)$:
\begin{equation}\label{equ: sigma hat}
\hat{\sigma}(t)
=
\hat{\sigma}(iT_s)
=-\bar{B}(iT_s)^{-1}\Phi^{-1}\mu(iT_s),
\end{equation}
where $\hat{\sigma} = [\hat{\sigma}_m^\top \ \hat{\sigma}_{um}^\top]^\top$. Note that the square matrix $\bar{B} \in \mathbb{R}^{6 \times 6}$ is invertible.

To see how the proposed estimation architecture works, consider the state prediction error dynamics via subtracting~\eqref{equ:quadrotor partial states uncertain dynamics} from~\eqref{equ:L1statepredictor}:
\begin{equation}\label{equ:prediction error dynamics}
    \dot{\tilde{z}}(t) = A_s \tilde{z}(t) +h(t)- \bar{B}(t)\sigma(t), \quad \tilde{z}(0)=0.
\end{equation}
Given that $h(t)$ is constant for $t \in [iT_s,(i+1)T_s)$, the solution of the linear system in \eqref{equ:prediction error dynamics} at $t=(i+1)T_s$ can be computed from $t=iT_s$:
\begin{align}\label{equ:closed-form solution to prediction error}
&\tilde{z}\left( (i+1)T_s \right) \nonumber 
\\ 
&=  \exp(A_s T_s) \tilde{z}(iT_s) + \left(\exp(A_s T_s) - I \right) A_s^{-1} h(iT_s) \nonumber 
\\
&\hspace{3.5mm}-\textstyle\int_{iT_s}^{(i+1)T_s} \exp(A_s((i+1)T_s-\tau)) \bar{B}(\tau) \sigma(\tau) \dd \tau.
\end{align}
As the state predictor \eqref{equ:L1statepredictor} replicates the system's structure in \eqref{equ:quadrotor partial states uncertain dynamics}, with the terms related to the uncertainties replaced by the estimation $h(t)$, a small prediction error implies an accurate estimation. 
Consequently, the design philosophy of the piecewise-constant adaptation law is to drive the prediction error~\eqref{equ:closed-form solution to prediction error} as small as possible. When the piecewise-constant adaptation law is executed at $t=iT_s$, we do not have enough information of $\sigma(t)$ for $t \in [iT_s,(i+1)T_s)$ to compute the integral term on the right-hand side of~\eqref{equ:closed-form solution to prediction error}. As a result, we design the piecewise-constant adaptation law to cancel the term $\exp(A_s T_s) \tilde{z}(iT_s)$ due to the initial condition of the sampling period. 
Using the adaptation law in~\eqref{equ:L1adaptationlaw} and plugging $h(iT_s)$ in \eqref{equ:closed-form solution to prediction error}, the prediction error at $t=(i+1)T_s$ takes the form
\vspace{-1.5mm}
\begin{align}\label{equ: prediction error final}
&\tilde{z}\left( (i+1)T_s \right) \nonumber \nonumber 
\\
&= -\textstyle\int_{iT_s}^{(i+1)T_s} \exp(A_s((i+1)T_s-\tau)) \bar{B}(\tau) \sigma(\tau) \dd \tau.
\end{align}
Moreover, following the same argument, we can see that the contribution of $\tilde{z}\left((i+1)T_s\right)$ to $\tilde{z}((i+2)T_s)$ will be eliminated at $t=(i+2)T_s$ by $h((i+1)T_s)$.

By setting the sampling time $T_s$ small enough (up to the hardware limit), one can achieve arbitrary accurate uncertainty estimation. Note that small $T_s$ leads to high adaptation gain $\Phi^{-1} \exp(A_sT_s)$ in~\eqref{equ:L1adaptationlaw}. When implemented, the high adaptation gain will amplify high-frequency components (resulting from sensor noise) in the estimates. Therefore, we filter the uncertainty estimates to prevent the high-frequency components from entering the control channel (as we will introduce next), which decouples the fast estimation from the control channel and protects the robustness of the system. The $\mathcal{L}_1$ adaptive control law $u_{ad}(t)$ only compensates for the matched uncertainty $\sigma_m(t)$ within the bandwidth $\omega$ of the LPF $C(s)$:

\begin{equation}\label{equ:L1controlinput}
u_{ad}(s)=-C(s)\hat{\sigma}_m(s),~~C(0)=1,
\end{equation}
where the signals are posed in the Laplace domain. 

Finally, the \ellone adaptive control input $u_{ad}(t)$ is integrated into the control loop such that the total control input is $u(t) = u_b (t) + u_{ad}(t)$, as shown in Fig~\ref{fig: L1 control framework}. The details of computing the baseline control input $u_b$ using geometric control are summarized in Section~\ref{sec: geometric controller}. We outline the implementation of the \ellone adaptive control in Algorithm~\ref{alg: L1 algorithm}, where we use a first-order low-pass filter with transfer function $C(s)=\omega / (s+\omega)$ as an example. To implement \ellone adaptive control, we only require a rough estimate of the nominal system's parameters (mass and inertia). We recommend constructing it in the sequence of state predictor, adaptation law, and low-pass filter. When tuning the state predictor and the adaptation law, we suggest disconnecting the $u_{ad}(t)$ from the control loop. 
If the state predictor and the adaptation law are constructed correctly, the predicted states $\hat{z}(t)$ should closely track the real states $z(t)$. Subsequently, one can connect $u_{ad}(t)$ into the control loop and tune the low-pass filter. We suggest beginning with a small bandwidth and slowly increasing the bandwidth to balance the performance and the robustness. 

\begin{algorithm}[ht]
	\caption{$\mathcal{L}_1$ adaptive control (in discrete-time notation with $k-1$ and $k$ indicating previous and current steps, respectively). Note that \textbf{Parameters} are designed by the user, and \textbf{Input} includes terms either from the state estimation or computation at the previous time step.}
	\label{algo: L1}
	\begin{algorithmic}[1]\label{alg: L1 algorithm}
	    \ENSURE  Hurwitz matrix $A_s$, sampling time $T_s$, and cutoff frequency $\omega$ of a first-order low-pass filter.
		\REQUIRE  $R(k)$, $z(k)$, $u_b(k-1)$, $z(k-1)$, $R(k-1)$, $\hat{z}(k-1)$,  $\hat{\sigma}_m(k-1)$, $\hat{\sigma}_{um}(k-1)$, $u_{ad}^{pre}(k-1)$.
		\STATE Update state predictor: \\
		    $\hat{z}(k)=  \hat{z}(k-1) + T_s \dot{\hat{z}}(k-1)$;
		\STATE Update state prediction error: $\tilde{z}(k) = \hat{z}(k) - z(k)$.
            \STATE Compute $h(k)$: \\$h(k)=(\exp(A_s T_s)-I)^{-1}A_s \exp(A_s T_s)\tilde{z}(k)$.
		\STATE Compute the estimated uncertainty: \\ \vspace{0.5mm}$\left[\begin{smallmatrix} \hat{\sigma}_m(k) \\ \hat{\sigma}_{um}(k) \end{smallmatrix}\right]= \hat{\sigma}(k) =-\bar{B}(R(k))^{-1}h(k)$.\vspace{1mm}
	\STATE{Filter the matched uncertainty estimate and negate the filtered signal to obtain \ellone adaptive control}: \\ $u_{ad}^{pre}(k) = e^{-\omega T_s}u_{ad}^{pre}(k-1) + (1-e^{-\omega T_s}) \hat{\sigma}_m(k)$, \\$u_{ad}(k) = -u_{ad}^{pre}(k)$.
		\RETURN \ellone adaptive control $u_{ad}(k)$.
	\end{algorithmic}
\end{algorithm}
\vspace{-4mm}
\section{Performance Analysis}\label{sec: performance analysis}
Next, we analyze the performance with the \ellone adaptive control. We show that the tracking errors are uniformly bounded with tunable bound size.
Note that we restrict the analysis to a first-order filter of the form $C(s)=\omega / (s+\omega)$. The results can be generalized to higher-order filters. Furthermore, we assume the states of the quadrotor can be measured accurately without noise. The proofs of Proposition~\ref{prop:bounds of Ps Ub and Uad} and~\ref{prop: estimation error bound}, and Theorem~\ref{thm: main theorem} in this section are provided in Appendix~\ref{apdx: proof of prop 3},~\ref{apdx: proof of prop 4}, and~\ref{apdx: proof of thm}, respectively.

Let \!$\bar{B}_F(R(t))\! \triangleq\! [B_F(R(t)) \ B^\bot_F(R(t))]$. 
\begin{proposition}\label{prop: bound of Bs.}
The functions $B(R(t))$, $\Bar{B}(R(t))$, $\Bar{B}^{-1}(R(t))$, $B_F(R(t))$, and $\Bar{B}_F(R(t))$ are bounded; i.e.  there exist constants $\Delta_B$, $\Delta_{\Bar{B}}$, $\Delta_{\Bar{B}^{-1}}$, $\Delta_{B_F}$, and $\Delta_{\Bar{B}_F}$ satisfying
\begin{equation}\label{bounds on Bs}
\begin{aligned}
\norm{B(R(t))} \leq& \Delta_B, & \norm{\Bar{B}(R(t))} \leq& \Delta_{\Bar{B}},\\
\norm{\Bar{B}^{-1}(R(t))}\leq& \Delta_{\Bar{B}^{-1}}, & \norm{B_F(R(t))} \leq& \Delta_{B_F},\\
\norm{\Bar{B}_F(R(t))} \leq& \Delta_{\Bar{B}_F},
\end{aligned}
\end{equation}
for all $t\geq0$.
\end{proposition}
\begin{proof}
Since $\det(R)=1$ and $J$ is a constant matrix, the five bounds in \eqref{bounds on Bs} hold for all $R(t)$ (and hence all $x(t)$) for $t \geq 0$.
\end{proof}
\begin{proposition}\label{prop:bounds of Ps Ub and Uad}
Let 
\begin{equation}\label{def: P1 P2 definition}
    \begin{aligned}
P_1(t) & \triangleq \begin{bmatrix}
-R(t)e_3 \cdot \left( \frac{c_1}{m}e_x(t) + e_v(t)  \right) \\ e_\Omega(t)+c_2 J^{-1} e_R(t) \end{bmatrix}^\top,\\
P_2(t) &\triangleq \frac{c_1}{m}e_x(t) + e_v(t).
    \end{aligned}
\end{equation}
There exist constants $c_3$, $\Delta_f$, $\Delta_{u_b}$, $\Delta_{u_{ad}}$ such that the following bounds hold for all $x(t) \in \mathcal{O}(x_d(t),\rho)$ and $t \geq 0$
\begin{equation}\label{bounds of P1 P2}
\norm{P_1} \leq  c_3 \rho, \ \norm{P_2} \leq  c_4 \rho, \ \norm{f(x(t))}\leq \Delta_f,
\end{equation}
\begin{equation}\label{bounds on ub and sigmahat}
    \norm{u_b}\leq \Delta_{u_b}, \quad \norm{u_{ad}}\leq \Delta_{u_{ad}},
\end{equation}
where $c_4$ is defined in~\eqref{unmatched special bound}.
\end{proposition}

For an arbitrarily chosen positive scalar $\epsilon$, let
\begin{subequations}
\begin{align}
\rho \triangleq& d(x(0),x_d(0))\sqrt{\overline{\gamma}/\underline{\gamma}}+\epsilon,     \label{def: definition of rho} \\
\label{def: phi1}
\phi_1 \triangleq& \Delta_f+ \Delta_{\bar{B}}\Delta_\sigma+\Delta_{B_F} (\Delta_{u_b}+\Delta_{\hat{\sigma}}), \\
\zeta_1(\omega) \triangleq& \frac{\Delta_{\sigma_m}}{\lvert \beta - \omega \rvert} + \frac{L_{\sigma_{m_t}} + L_{\sigma_{m_x}}\phi_1}{\beta \omega},\label{equ: zeta 1}\\
\zeta_2(A_s) \triangleq&   2\sqrt{6}\Delta_{\Bar{B}}\Delta_{\Bar{B}^{-1}}\left( \phi_1 L_{\sigma_x} +  L_{\sigma_t}\right) \nonumber \\ &+ \sqrt{6}\Delta_{\Bar{B}^{-1}}\left(  L_B + \abs{\lambda_{m}(A_s)}\Delta_{\Bar{B}}\right)\Delta_\sigma, \label{equ: zeta 2}\\
\label{equ: zeta 3}
\zeta_3(\omega)\triangleq&\Delta_\sigma \omega, \\
\label{equ: zeta 4}
\zeta_4(A_s,\omega) \triangleq& \zeta_2(A_s) + \zeta_3(\omega).
\end{align}
\end{subequations}
Choose the sampling time $T_s$ and the bandwidth $\omega$ of the low-pass filter $C(s)$ to satisfy the following conditions:
\begin{subequations}
\begin{align}\label{equ:condition on filter bandwidth}
\underline{\gamma}\rho^2 >& c_3 \rho \zeta_1(\omega) + c_4\rho\Delta_{\sigma_{um}} + V_0, \\
\label{equ: condition on sampling rate}
    T_s \leq& \frac{\underline{\gamma}\rho^2- c_3 \rho \zeta_1(\omega) - c_4\rho \Delta_{\sigma_{um}}-V_0}{\zeta_4(A_s,\omega)}.
\end{align}
\end{subequations}

\begin{remark}\label{remark on existence of qualified bandwidth and sampling time}
Based on the definition of $\rho$ in~\eqref{def: definition of rho} and the bounds of the Lyapunov function in~\eqref{equ: lyapunovboundfinal}, the inequality $\underline{\gamma}\rho^2 > V_0$ holds. Furthermore, since $\zeta_1(\omega)$ converges to zero as $\omega$ increases, under Assumption~\ref{asump: unmatched uncertainty assumption}, conditions~\eqref{equ:condition on filter bandwidth} and~\eqref{equ: condition on sampling rate} can always be satisfied by choosing a sufficiently large $\omega$ and a sufficiently small $T_s$. 
\end{remark}
\begin{proposition}\label{prop: estimation error bound}
Let Assumption~\ref{aspt: main assumption} hold. If  $x(t) \in \mathcal{O}(x_d(t),\rho)$ for all $t\ge 0$, then 
\begin{equation}
    \norm{\sigma(t,x(t))-\hat{\sigma}(t,x(t))} \leq \left\{ \begin{aligned} & \Delta_\sigma,  && \forall 0\leq t < T_s,\\& \zeta_2(A_s)T_s,  && \forall t \geq T_s. 
    \end{aligned} \right.
\end{equation}
\end{proposition} 

\begin{remark}\label{rmk: estimation error bound}
Notice that for $t \geq T_s$, the estimation error converges to zero as the sampling time $T_s$ goes to zero. 

\end{remark}

Finally, we are ready to state our main theorem.
\begin{theorem}\label{thm: main theorem}
Let Assumptions~\ref{aspt: main assumption} and~\ref{asump: unmatched uncertainty assumption} and the conditions in~\eqref{equ:condition on filter bandwidth} and~\eqref{equ: condition on sampling rate} hold. 
Then using the geometric control law in~\eqref{eq:thrust control} and~\eqref{eq:torque control}, and the $\mathcal{L}_1$ adaptive control law in~\eqref{equ:L1controlinput}, we have 
\begin{equation}\label{equ: main theorem }
    x(t) \in \mathcal{O}(x_d(t),\rho), \quad \forall t \geq 0.
\end{equation}
Moreover, the closed-loop system's state $x(t)$ in~\eqref{equ:uncertain system with full state and with L1} is uniformly ultimately bounded, i.e., for any $t_1 > 0$, we have
\begin{equation}
 x(t) \in \mathcal{O}\left(x_d(t),\mu(\omega,T_s,t_1)\right) \subset \mathcal{O}(x_d(t),\rho), \quad \forall t \geq t_1 >0,
\end{equation}
where the ultimate bound is defined as
\begin{align}\label{ultimate bound}
\mu&(\omega,T_s,t_1) \triangleq \nonumber \\
&\sqrt{\frac{e^{-\beta t_1}V_0 + c_3 \rho \zeta_1(\omega) + \zeta_4(A_s,\omega)T_s + c_4\rho\Delta_{\sigma_{um}}}{\underline{\gamma}}}.
\end{align}
\end{theorem}

\begin{remark}
Following the $\mathcal{L}_1$ adaptive control analysis in~\cite{hovakimyan2010L1}, one can also derive the performance bounds for the adaptive control input $u_{ad}$, which are essential for robustness guarantees of $\mathcal{L}_1$ adaptive controller. Due to page limitations, these derivations are not pursued in this paper.
\end{remark}
\begin{remark}
It is evident that the size of the uniform bound in~\eqref{equ: main theorem },  $\rho$, is lower bounded by the initial distance $d(x(0),x_d(0))$ and the positive scalars $\underline{\gamma}$ and $\overline{\gamma}$ (associated with the Lyapunov function of the nominal dynamics) from the definition in~\eqref{def: definition of rho}. Unless one redesigns the Lyapunov function that can tighten the bound in~\eqref{equ: lyapunovboundfinal}, the only way this bound $\rho$ can be further reduced is when the planner (which provides the desired state $x_d(0)$) can reduce $d(x(0),x_d(0))$. 
\end{remark}

\begin{remark}
Theorem~\ref{thm: main theorem} also provides a uniform ultimate bound $\mu(\omega,T_s,t_1)$ as shown in~\eqref{ultimate bound}. Since $e^{-\beta t_1}V_0$ decreases exponentially and  $\zeta_1(\omega)$ (defined in~\eqref{equ: zeta 1}) converges to zero as $\omega$ approaches infinity, if we do not consider the unmatched uncertainty (i.e., $\Delta_{\sigma_{um}} = 0$), then the bound $\mu(\omega,T_s,t_1)$ can be rendered arbitrarily small in finite time $0 < t_1 \leq t < \infty$ with a large $\omega$ and a small $T_s$. This feature is critical in controlling quadrotors through tight and cluttered environments: the bandwidth $\omega$ of $C(s)$ is the tuning knob of the trade-off between the performance and the robustness, while how small $T_s$ can be is subject to hardware limitations. 
The role of the low-pass filter $C(s)$ in the $\mathcal{L}_1$ adaptive control architecture is to decouple the control loop from the estimation loop~\cite{hovakimyan2010L1}. 
Thus, increasing the bandwidth $\omega$ to get a tighter bound will lead to $u_{ad}$ behaving as a high-gain signal, possibly sacrificing desired robustness levels~\cite{cao2010stability}. Therefore, this trade-off must always be taken into account when designing the low-pass filter.
\end{remark}

\begin{remark}
We want to highlight that the computational load for running $\mathcal{L}_1$Quad is fixed in each sampling interval. Thus, the sampling interval $T_s$ can be reduced as much as possible, only subject to hardware limitations (e.g., CPU power and time for performing other computations). It is worth noting that $\mathcal{L}_1$Quad holds a light computational load, involving only matrix summation and multiplication, along with matrix exponential of diagonal matrices. While we do employ a matrix inverse $\bar{B}^{-1}$ as shown in~\eqref{equ: sigma hat}, this matrix inverse has an explicit expression, allowing for fast computation without relying on numerical procedures.
\end{remark}

\section{Experimental results}\label{sec:experiments}
We demonstrate the performance of $\mathcal{L}_1$Quad on a custom-built quadrotor shown in Fig~\ref{fig:quadrotor with reference frames}. The quadrotor weighs 0.63~kg with a 0.22~m diagonal motor-to-motor distance. We use four T-Motor F60 2550KV BLDC motors with T5150 tri-blade propellers. The motors are driven by a Lumenier Elite 60A 4-in-1 ESC and a 4S Lipo battery.
The quadrotor is controlled by a Pixhawk 4 mini flight controller running the ArduPilot firmware. We modify the firmware
\footnote{Our firmware can be accessed from https://github.com/sigma-pi/L1Quad}
to enable the geometric controller and the $\mathcal{L}_1$ adaptive control, which both run at 400~Hz on the Pixhawk. Position feedback is provided by 9 Vicon V16 cameras, covering a motion capture volume of $7 \times 5 \times 3$ m$^3$. A ground station routes the Vicon measurements to the Pixhawk through Wifi telemetry at 50~Hz. We use ArduPilot's EKF to fuse the Vicon measurements with IMU readings onboard. A list of parameters in the experiments is shown in Table~\ref{tb: parameters in experiments}. We use the terms ``$\mathcal{L}_1$ on'' and ``$\mathcal{L}_1$ off'' to indicate that the system is controlled with and without the $\mathcal{L}_1$ adaptive control, respectively.
\setlength{\tabcolsep}{5pt} 
\renewcommand{\arraystretch}{1} 
  \captionsetup{
	skip=5pt, position = bottom}
\begin{table}[H]
	\centering
	\small
	\caption{Parameters in experiments.}
	\begin{tabular}{llcll}
		\toprule[1pt]
		 param. & value & & param. & value \\
		\cmidrule{1-2} \cmidrule{4-5} 
        $m$  & 0.62 kg & & $J$ & $10^{-3}\diag{3.0 , 1.8 , 3.2}$ kgm$^2$ \\
        g & 9.81 m/s$^2$ & & $K_p$ & $\diag{14 , 15 , 15}$\\
        $T_s$ & 0.0025 s & & $K_v$ & $\diag{1.5 , 0.9 , 1.1}$ \\
        $\omega_c$ ($f$) & 30  rad/s && $K_R$ & $10^{-1}\diag{5.5 , 3.5 , 1.5} $ \\
        $\omega_c$ ($M$) & 5, 15 rad/s & & $K_{\Omega}$ & $10^{-2}\diag{3.5 , 3 , 0.4} $  \\
        & && $A_s$ & $-\diag{5 , 5 , 5 , 10 , 10 , 10}$ \\
		\bottomrule[1pt]
	\end{tabular}\label{tb: parameters in experiments}
\end{table}

\begin{remark}\label{rmk: parameter tuning}
We tune the control parameters $K_p$, $K_v$,  $K_R$, and $K_{\Omega}$ in Table~\ref{tb: parameters in experiments} to ensure that the quadrotor hovers and tracks a circular trajectory with 1 m radius and 0.5 m/s linear speed, achieving a {root-mean-square error (RMSE)} of less than 3 cm. We then tune the bandwidth of the low-pass filter $\omega_c$ to balance the performance and the robustness. These parameters tuned in nominal scenarios are then employed in all the experiments presented below. We do not retune the parameters for any specific case. 
\end{remark}

\normalsize
\subsection{Uncertainty estimation and compensation}\label{subsec: injected unc}

In this section, we demonstrate the $\mathcal{L}_1$Quad's capability of estimating and compensating for uncertainties by comparing the artificially injected uncertainty $\sigma_{inj}(t) \in \mathbb{R}^4$, with the estimated uncertainty computed by the proposed architecture. We command the quadrotor to hover at $1$~m altitude and inject uncertainty to the thrust channel by setting the first element of $\sigma_{inj}(t)$ to $0.6  \sin(2 \pi t) + 0.6  \sin(\pi t)$ starting from $t=2$ as shown in Fig~\ref{fig:injected unc to thrust}. The $\mathcal{L}_1$ adaptive control input $u_{ad}(t)$ compensates for the pre-existing uncertainty for $t \in [0,2)$ and the sum of the pre-existing uncertainty and the injected uncertainty for $t \in [2, 16]$. To make sure the comparison is between the injected uncertainty and its estimation, we shifted the plot of $u_{ad}$ vertically to eliminate the effect of the pre-existing uncertainty. As shown in Fig~\ref{fig:injected unc to thrust}, the shifted $-u_{ad}$ closely follows the injected uncertainty $\sigma_{inj}$, which implies that $\mathcal{L}_1$Quad  accurately estimates and effectively compensates for uncertainties. 
\begin{figure}
     \begin{subfigure}[b]{\columnwidth}
         \centering
         \includegraphics[width=\columnwidth]{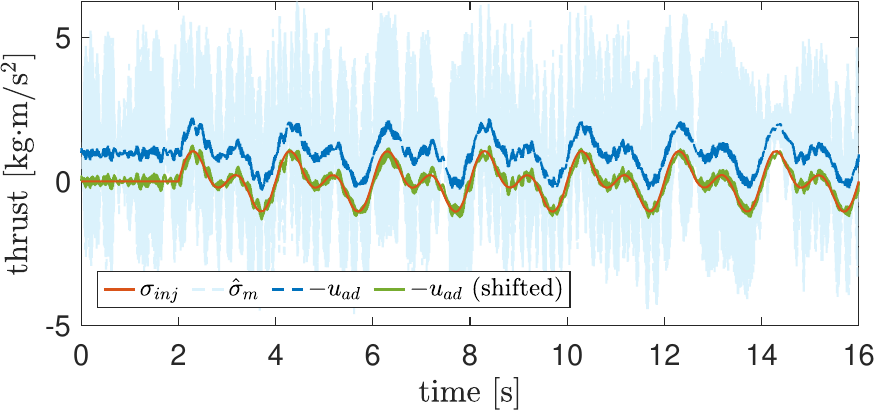}
         \caption{Injected uncertainty $\sigma_{inj}$ with estimation $\hat{\sigma}_m$ and compensation $u_{ad}$, all on the thrust channel.}\vspace{3mm}
        \label{fig:injected unc to thrust}
     \end{subfigure}
     \\
      \begin{subfigure}[b]{0.49\columnwidth}
         \centering
         \includegraphics[width=\columnwidth]{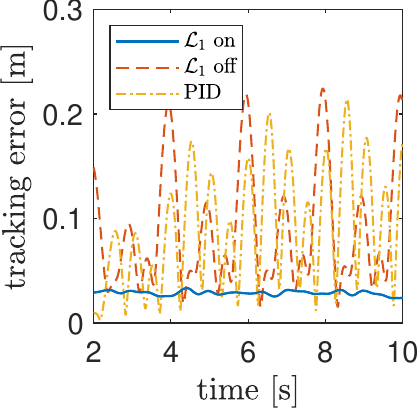}
         \caption{Comparison of the tracking error between $\mathcal{L}_1$ on, $\mathcal{L}_1$ off, and PID.}
        \label{fig:injected unc tracking error}
     \end{subfigure}
     \begin{subfigure}[b]{0.49\columnwidth}
         \centering
         \includegraphics[width=\columnwidth]{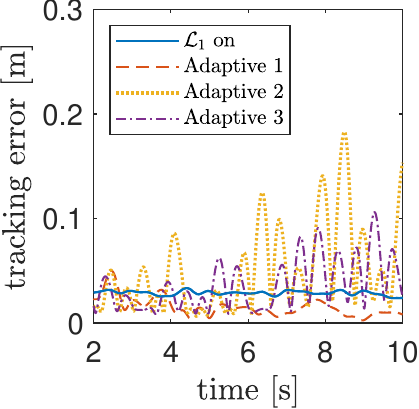}
         \caption{Comparison of the tracking error between $\mathcal{L}_1$ on, Adaptive 1, 2, and 3.}
        \label{fig:injected unc tracking error2}
     \end{subfigure}
     \caption{Results for the experiments with injected uncertainties on the thrust channel. Details about Adaptive 1, 2, and 3 are provided in Section~\ref{subsec: injected unc}.}
     \vspace{-4mm}
\end{figure}
Figure~\ref{fig:injected unc tracking error} and~\ref{fig:injected unc tracking error2} show the tracking error of the system with the $\mathcal{L}_1$ adaptive control on and off, where the former can keep a consistently small tracking error owing to accurate estimation and effective compensation of the uncertainties.
We also compare the performance of the existing geometric PID controller~\cite{goodarzi2013geometric} and the nonlinear adaptive controller~\cite{goodarzi2015geometric}, both offering uncertainty compensation by design. These two controllers use the same values of $K_p$, $K_v$, $K_R$, and $K_{\Omega}$ of the geometric controller in Table~\ref{tb: parameters in experiments}.
The PID controller uses parameters $k_i=3$, $c_1=2.5$, $k_I=0.2$, and $c_2=2$ following the notation in~\cite{goodarzi2013geometric}. For the adaptive controller~\cite{goodarzi2015geometric}, uncertainties are captured using a linearly parameterized approach, e.g., the uncertainties on the translational dynamics are parameterized by $\mathbb{W} \theta$ with $\mathbb{W}$ being the collection of basis functions (feature) and $\theta$ being the weights. 
Note that designing the basis functions in $\mathbb{W}$ is a challenging problem. In this experiment, we select $\mathbb{W}=[\diag{1,1,1},\ \mathcal{B}^\ast]$, where $\mathbb{W} \in \mathbb{R}^{3 \times 4}$, the first three columns of $\mathbb{W}$ capture the constant uncertainties on the $x$-, $y$-, and $z$-axis, and the last column has function $\mathcal{B}^\ast=[0,\ 0,\ \mathcal{B} ]^\top$, where $\mathcal{B}$ represents a scalar function that captures the time-varying ones on the $z$-axis. We test the performance of the adaptive controller~\cite{goodarzi2015geometric} under three basis function $B_1^\ast$, $B_2^\ast$, and $B_3^\ast$, named \textit{Adaptive~1}, \textit{Adaptive~2}, and \textit{Adaptive~3}, respectively, in Fig~\ref{fig:injected unc tracking error2}. 
Specifically, the third element of $\mathcal{B}_1^\ast$, $(\mathcal{B}^\ast_1)_3 = \sin(2 \pi  t) + \sin(\pi  t)$ provides a basis function to fully capture the injected uncertainty; $(\mathcal{B}^\ast_2)_3 = \sin(6t) + \sin(3  t)$ is the function with frequency off (by 4.5\%); and $(\mathcal{B}^\ast_3)_3 = \sin(2 \pi  t + \pi/8) + \sin(\pi t + \pi/8)$ is the function with phase off\footnote{We also test the basis function with a larger phase shift of $\pi/4$. 
The quadrotor crashed due to the phase lag, and this case is not reflected in Fig~\ref{fig:injected unc tracking error2}.}. 
\textit{Adaptive~1} achieves the smallest tracking error with the precise function $\mathcal{B}_1^\ast$ matching the injected uncertainty. Note that the precise basis function (with identical frequency and phase to the uncertainties) is unrealistic in practice. 
The imprecise basis functions $\mathcal{B}^\ast_2$ and $\mathcal{B}^\ast_3$ lead to worse tracking performance than $\mathcal{B}^\ast_1$. The degraded performance illustrates the challenges in choosing the appropriate basis functions in the traditional nonlinear adaptive control that requires a parametric structure for uncertainties: a small deviation of frequency or phase of the true uncertainty leads to significant performance degradation.
Overall, the existing adaptive controller~\cite{goodarzi2015geometric} performs better than the PID controller~\cite{goodarzi2013geometric}, whose performance is close to that of the geometric controller ($\mathcal{L}_1$ off). The integral term in PID takes effect after sufficiently integrating the uncertainty-induced tracking error, which performs poorly in this case with the fast-changing injected uncertainty.

\subsection{Slosh payload}
Slosh payload is a challenging scenario for the quadrotor's tracking control: the sloshing liquid constantly changes its distribution within the container, causing constant changes in the vehicle's center of mass and moment of inertia. Due to the high dimensions of the fluid dynamics, modeling the liquid distribution as a function of the quadrotor's state is impossible for real-time computation.
In this experiment, we attach a bottle to the bottom of the quadrotor (with the bottle's longitudinal direction aligned with the quadrotor's body-$x$ direction) and fill the bottle half-full. The total added weight is 0.317~kg, approximately 50\% of the quadrotor's weight. We command the quadrotor to track a circular trajectory with a 1 m radius and 2.5 m/s linear speed. 
To track such a trajectory, the quadrotor needs to tilt on the roll or pitch axis for a maximum of 45 degrees, leading to drastic sloshing of the liquid payload, as shown in Fig~\ref{fig:slosh payload snapshot}. 
The tracking error comparison between the $\mathcal{L}_1$ on and off is shown in Fig~\ref{fig:slosh payload tracking error}. Notably, since the slosh motion majorly happens on the $x$-axis of the quadrotor, there is an apparent improvement in tracking performance on the $x$-axis by the $\mathcal{L}_1$ adaptive control. For the total position tracking error, the $\mathcal{L}_1$ adaptive control maintains a consistently smaller tracking error than when it is off.
\begin{figure}
        \begin{subfigure}[b]{\columnwidth}
         \centering
         \includegraphics[width=\columnwidth]{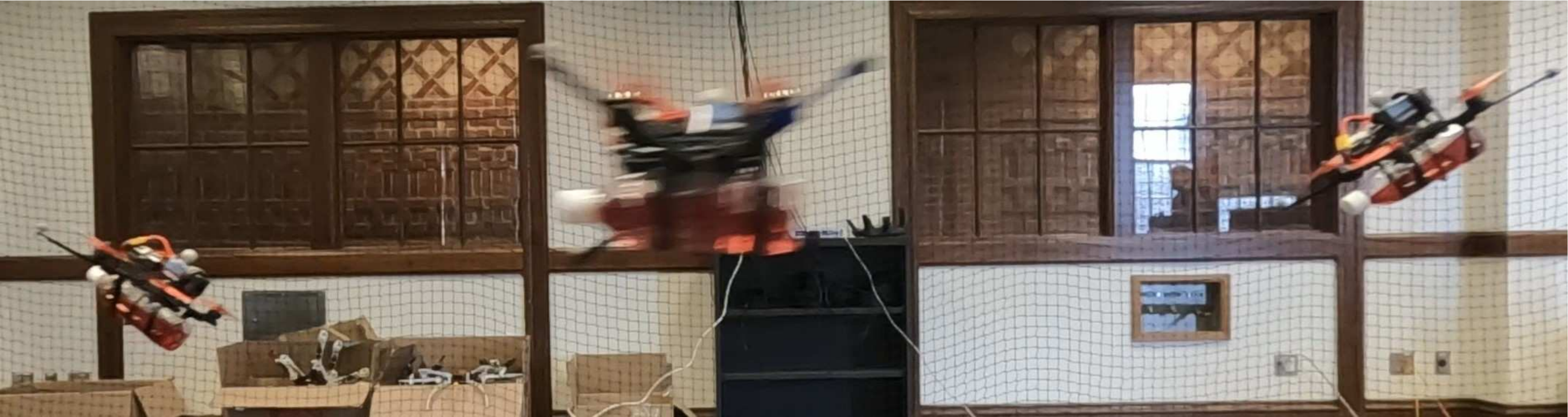}
         \caption{Snapshots of the quadrotor during the slosh payload test. Notice the tilted liquid surface when the quadrotor is at the left and right positions.}
        \label{fig:slosh payload snapshot}\vspace{3mm}
     \end{subfigure}
     \\
         \begin{subfigure}[b]{\columnwidth}
         \centering
         \includegraphics[width=\columnwidth]{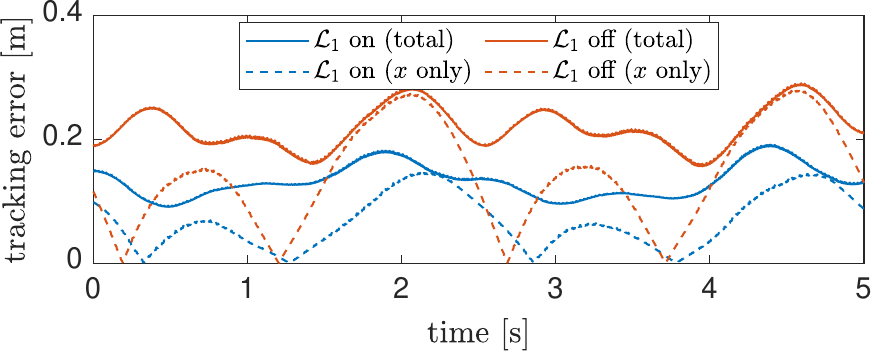}
         \caption{Position tracking error with the slosh payload for $\mathcal{L}_1$ on and off. The plots correspond to two rounds of the circular trajectory.}
        \label{fig:slosh payload tracking error}
     \end{subfigure}
     \caption{Slosh payload experiment.}
     \label{fig:slosh payload}
     \vspace{-6mm}
\end{figure}

\subsection{Tunnel}
In this experiment, we command the quadrotor to fly through a tunnel with 0.6 m (H) $\times$ 0.6 m (W) $\times$ 1.2 m (D). The tunnel is built to mimic the flying condition in confined space, e.g., a cave or mine passageway, where the quadrotor itself induces and therefore experiences complex airflow. We use Acrylic sheets for their transparency, which permits Vicon localization inside the tunnel. As shown in Fig~\ref{fig:tunnel snapshot}. the quadrotor flies from left to right following the trajectory $-1.2\cos(1/1.2t)$ with altitude of 0.3 m. Comparison of the tracking errors of $\mathcal{L}_1$ on, $\mathcal{L}_1$ off, and PID is shown in Fig~\ref{fig:tunnel tracking error}. With $\mathcal{L}_1$ on, the quadrotor achieves the minimum tracking error (within 0.03 m) while staying consistently smaller than with $\mathcal{L}_1$ off or PID. The latter two cases have the maximum tracking error that triples the one with $\mathcal{L}_1$ on, favoring $\mathcal{L}_1$ adaptive control in the environment with confined space.
\begin{figure}[H]
    \centering
    \begin{subfigure}[b]{\columnwidth}
         \centering
         \includegraphics[width=\columnwidth]{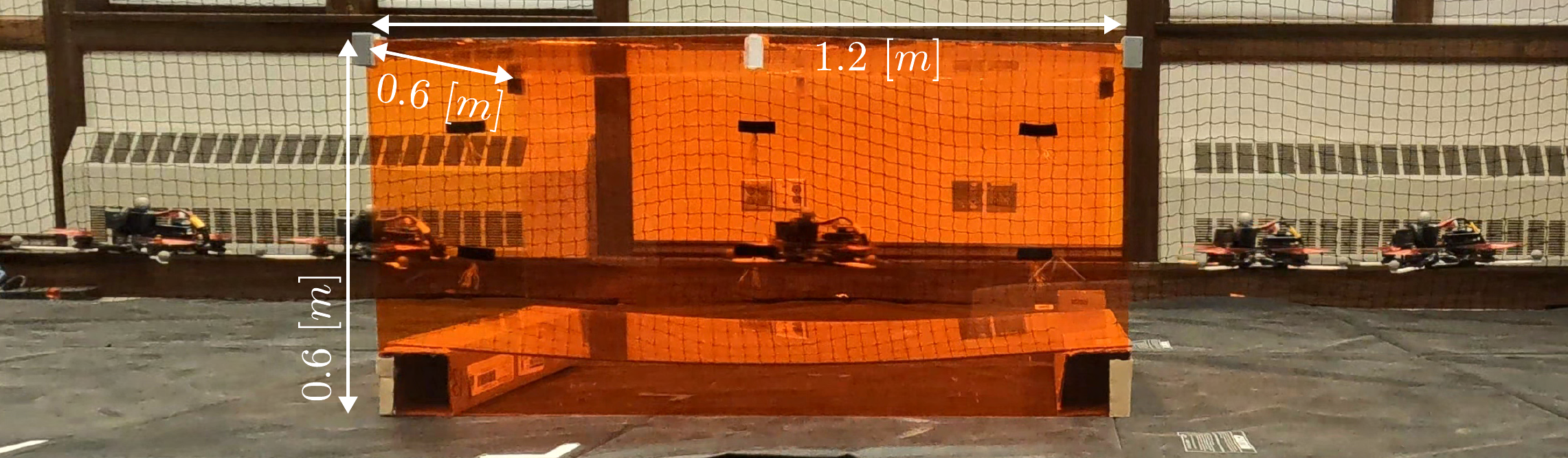}
         \caption{Snapshot of the tunnel experiment with $\mathcal{L}_1$ on.}
        \label{fig:tunnel snapshot} \vspace{3mm}
    \end{subfigure}
        \begin{subfigure}[b]{\columnwidth}
         \centering
         \includegraphics[width=\columnwidth]{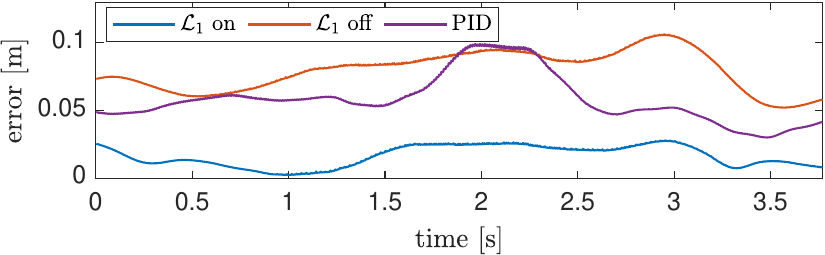}
         \caption{Tracking error of the compared controllers. The quadrotor enters and leaves the tunnel at around 1.26 s and 2.51 s, respectively.}
        \label{fig:tunnel tracking error}
    \end{subfigure}
    \caption{Results for the tunnel experiment.}\vspace{-2mm}
\end{figure}

\subsection{Chipped propeller}\label{subsec: chipped prop}
In this experiment, we chipped the tip of one propeller (shown in Figure~\ref{fig:chipped prop illustration}). We removed 1/4 of the propeller, which reduces the generated thrust from this propeller.
\begin{figure}[H]
    \centering
    \begin{subfigure}[b]{0.49\columnwidth}
        \includegraphics[width = 0.9\columnwidth]{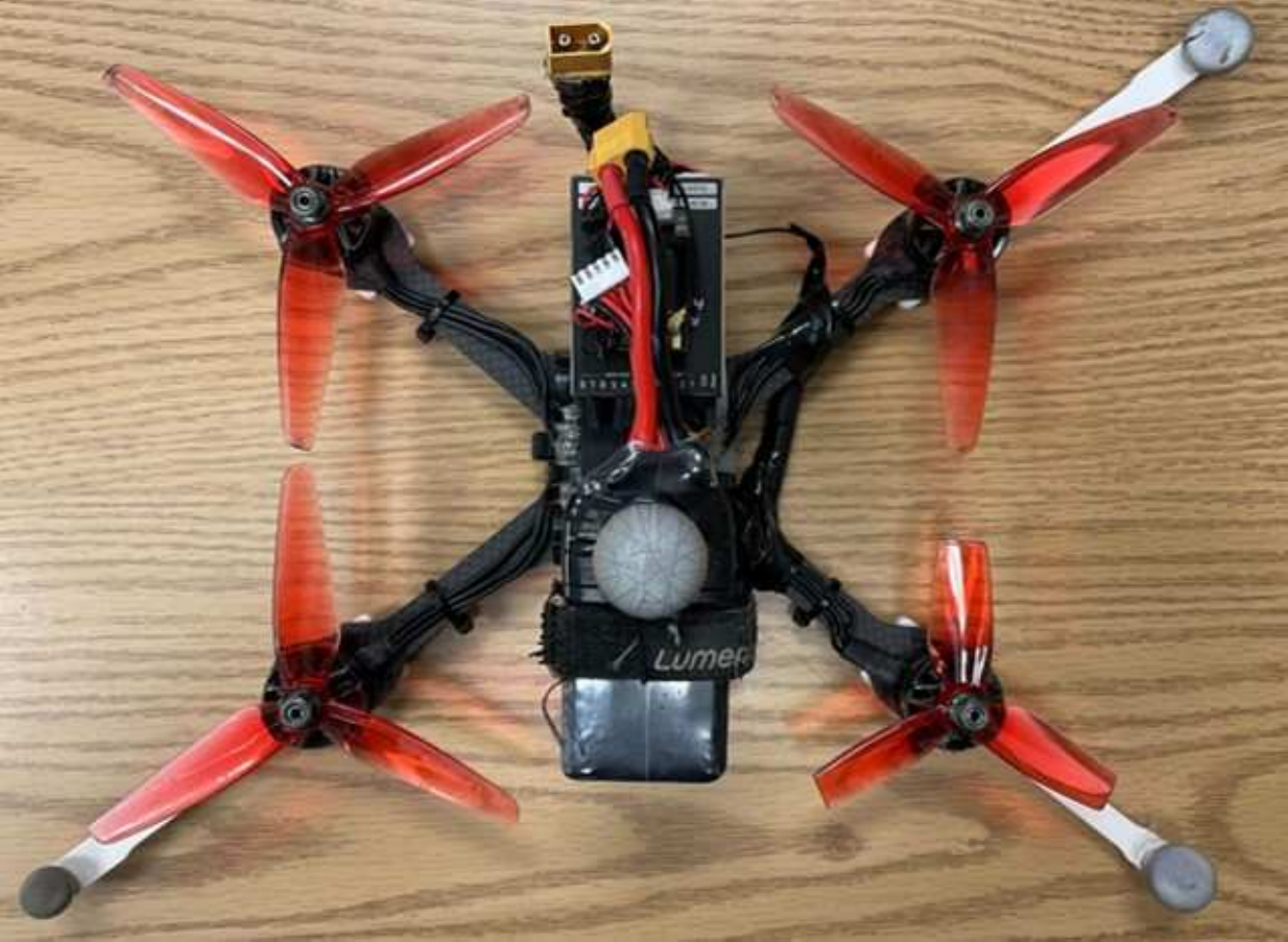}
        \caption{}
    \label{fig:chipped prop illustration}
    \end{subfigure}
    \begin{subfigure}[b]{0.49\columnwidth}
        \includegraphics[width=0.95\columnwidth]{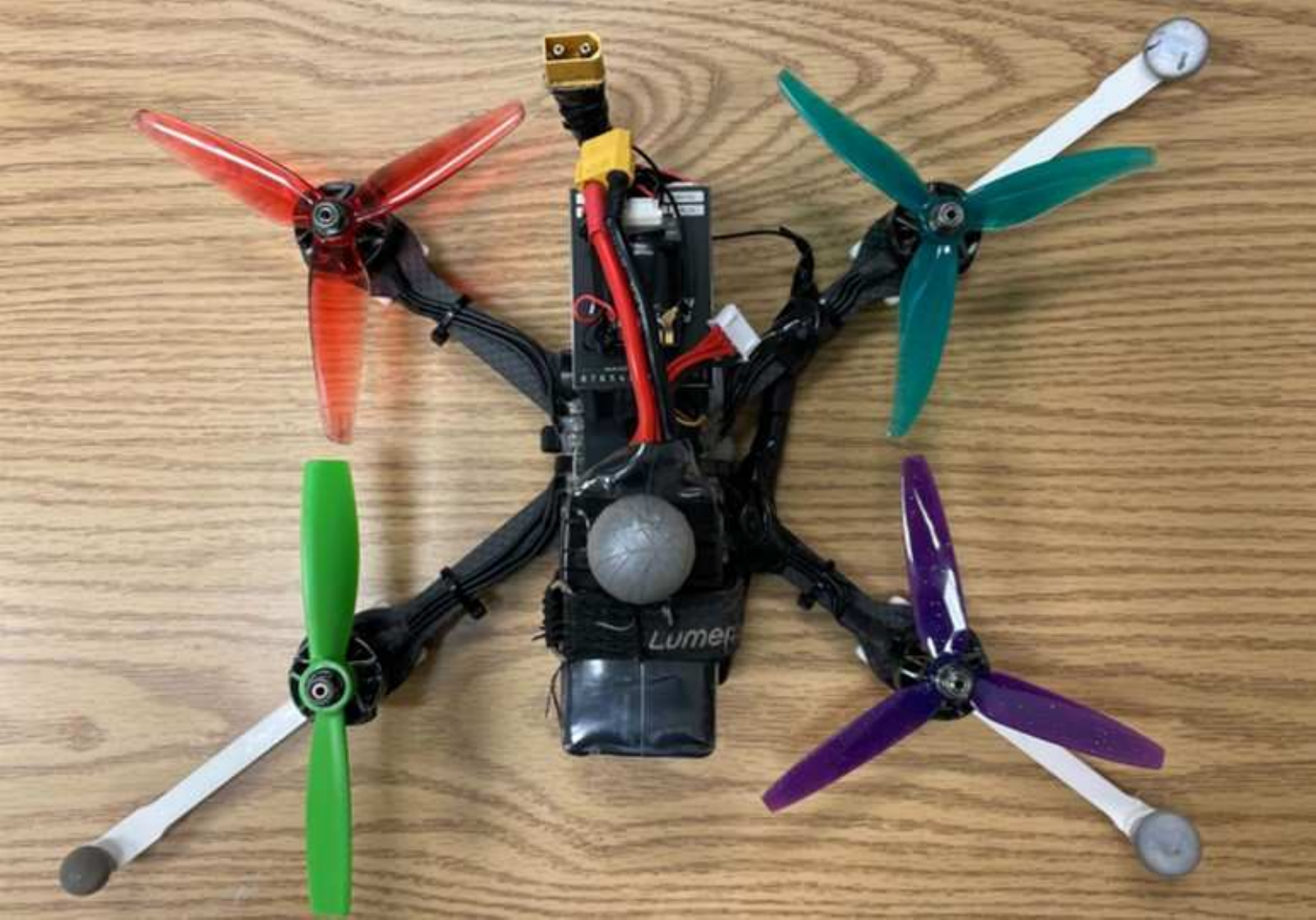}
        \caption{}
    \label{fig:mixed prop illustration}
    \end{subfigure}
    \caption{Quadrotors with (a) a chipped propeller and (b) mixed propellers. The mixed propellers are T-Motor T5150 Tri-Blade Propeller (upper left), HQProp 5$\times$4.5RG Bullnose Propeller (lower left), Gemfan 51466 V2 Tri-Blade Propeller (upper right), and Gemfan 513D Tri-Blade 3D Propeller (bottom right).}\vspace{-3mm}
\end{figure}

The quadrotor takes off to 1 m altitude from $-3.8$ to $-2$~s, during which the $\mathcal{L}_1$ adaptive control is turned off. The quadrotor is supposed to hover after $-2$ s. However, due to the thrust loss at the chipped propeller, the quadrotor cannot stabilize itself in the hovering position. 
This instability is reflected by the oscillations observed from $-2$ to $0$ s in Fig~\ref{fig:chipped prop L1 off to on}. At 0 s, we turn the $\mathcal{L}_1$ adaptive control on. Consequently, the quadrotor can stabilize itself to hover owing to $\mathcal{L}_1$Quad's compensation to the thrust generated by the chipped propeller. 
We further conduct the test that transitions the quadrotor from $\mathcal{L}_1$ on to $\mathcal{L}_1$ off, with results shown in Fig~\ref{fig:chipped prop L1 on to off}. When the $\mathcal{L}_1$ adaptive control is turned off at $0$~s, the quadrotor once again begins to oscillate.
\begin{figure}[H]
\vspace{-2mm}
     \begin{subfigure}[b]{\columnwidth}
         \centering         \includegraphics[width=\columnwidth]{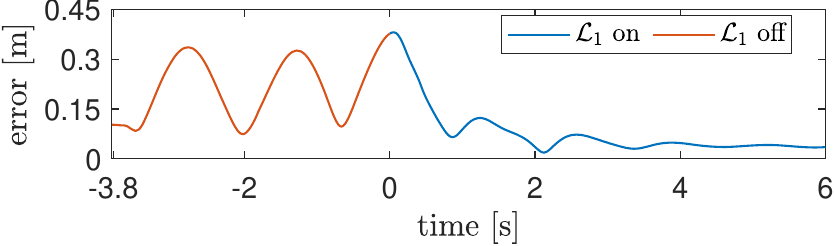}
         \caption{Switching $\mathcal{L}_1$ off to $\mathcal{L}_1$ on during hover.}
        \label{fig:chipped prop L1 off to on}
     \end{subfigure}
     \\
      \begin{subfigure}[b]{\columnwidth}
         \centering
         \includegraphics[width=\columnwidth]{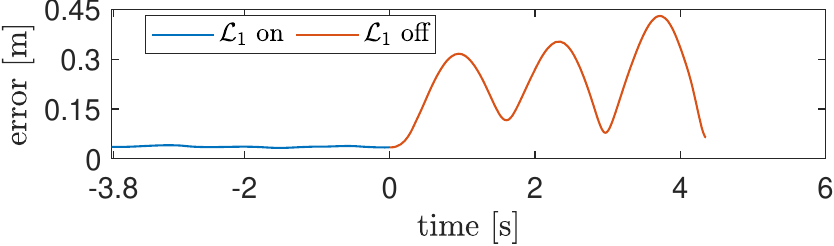}
         \caption{Switching $\mathcal{L}_1$ on to $\mathcal{L}_1$ off during hover.}
        \label{fig:chipped prop L1 on to off}
     \end{subfigure}
     \caption{Transient tracking performance when switching between $\mathcal{L}_1$ on and off with one chipped propeller.}
     \label{fig:chipped prop switching between L1 on and off}
     \vspace{-1mm}
\end{figure}

\setcounter{figure}{10}
\begin{figure*}[b]
     \begin{subfigure}[b]{\columnwidth}
         \centering
         \includegraphics[width=\columnwidth]{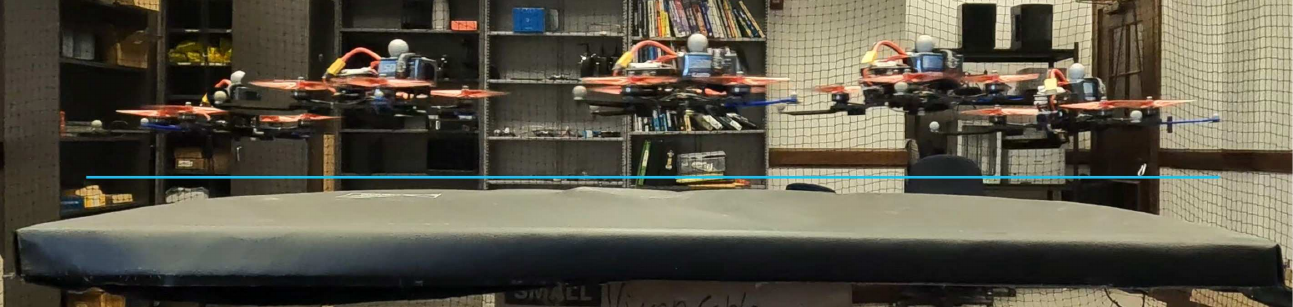}
         \caption{$\mathcal{L}_1$ off}
        \label{fig:ground effect L1 off illustration}
     \end{subfigure}
     \hfill
      \begin{subfigure}[b]{\columnwidth}
         \centering
         \includegraphics[width=\columnwidth]{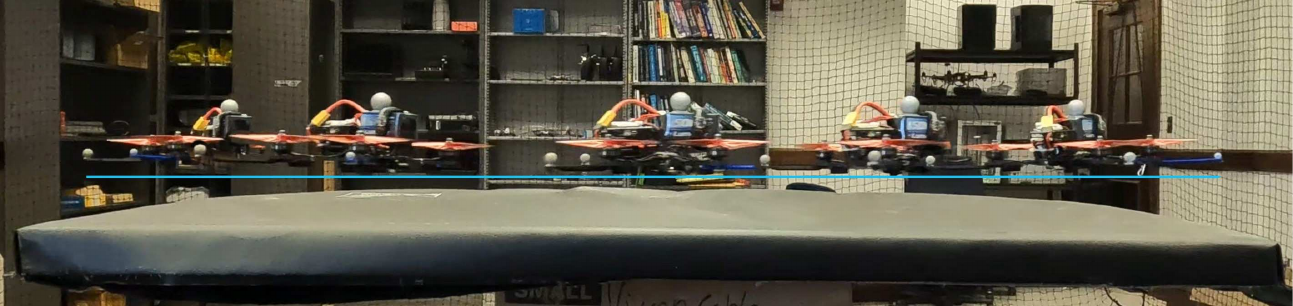}
         \caption{$\mathcal{L}_1$ on}
        \label{fig:ground effect L1 on illustration}
     \end{subfigure}
     \caption{The quadrotor flies over the ground effect zone. A blue line is supplied for comparing the altitude tracking performance.}
     \label{fig:ground effect illustration}\vspace{-4mm}
\end{figure*}

Moreover, we conduct the experiment to have the quadrotor (with the chipped propeller) follow a circular trajectory with a 1 m radius and 0.5 m/s linear speed. The actual trajectory is compared to the desired trajectory, as shown in Fig~\ref{fig:chipped prop 0.5 m/s tracking}. The RMSE is 0.045 m. We show the commanded pulse-width modulation (PWM) for each motor in Fig~\ref{fig:chipped prop motor PWM}, where motor 4 (with the chipped propeller) counteracts the thrust loss by commanding a higher PWM than the other motors. We also test the same flight with $\mathcal{L}_1$ off, in which the quadrotor crashed due to thrust loss from the chipped propeller.
\setcounter{figure}{7}
\begin{figure}[H]
\vspace{-4mm}
     \begin{subfigure}[b]{\columnwidth}
         \centering
         \includegraphics[width=\columnwidth]{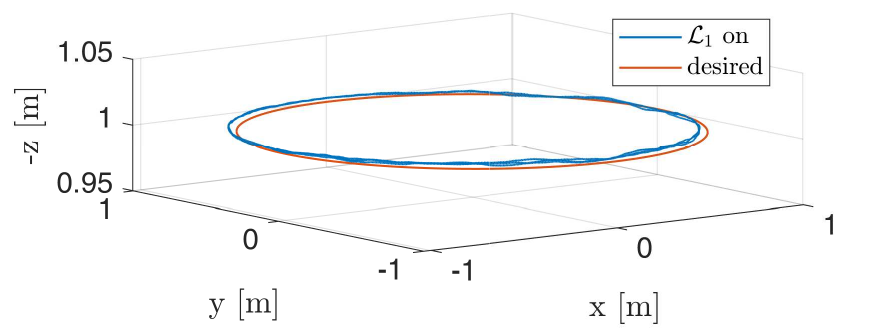}
         \caption{Comparison between the real trajectory and desired trajectory.}
        \label{fig:chipped prop 0.5 m/s tracking}
     \end{subfigure}

      \begin{subfigure}[b]{\columnwidth}
         \centering
         \includegraphics[width=\columnwidth]{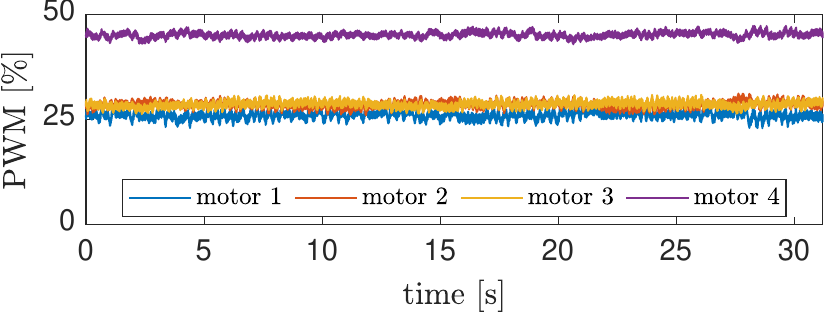}
         \caption{Pulse-width modulation (PWM) for the four motors.}
        \label{fig:chipped prop motor PWM}
     \end{subfigure}
     \caption{Results for the quadrotor with a chipped propeller circulating at 0.5 m/s with 1 m radius.}\vspace{-2mm}
     \label{fig:chipped prop 0.5 m/s results}
     \vspace{-1mm}
\end{figure}
\subsection{Mixed propellers}\label{subsec: mixed prop}
In this experiment, we apply four different propellers on the quadrotor (inspired by the experiment in \cite{saviolo2023active}). Figure~\ref{fig:mixed prop illustration} shows the propeller configuration. Since each propeller has its own thrust coefficient, the mixed propellers lead to unbalanced scaling of thrust from each propeller. We command the quadrotor to hover, and the tracking performance is shown in Fig~\ref{fig:mixed prop tracking}. We start with $\mathcal{L}_1$ on. When the $\mathcal{L}_1$ is turned off at 5 s, the tracking performance deteriorates instantly, and the quadrotor starts to oscillate with increasing magnitude.
We switch $\mathcal{L}_1$ back on at 26 s, and the $\mathcal{L}_1$Quad stabilizes the quadrotor with fast transient and small tracking error.

\begin{figure}[H]
    \centering
    \includegraphics[width=\columnwidth]{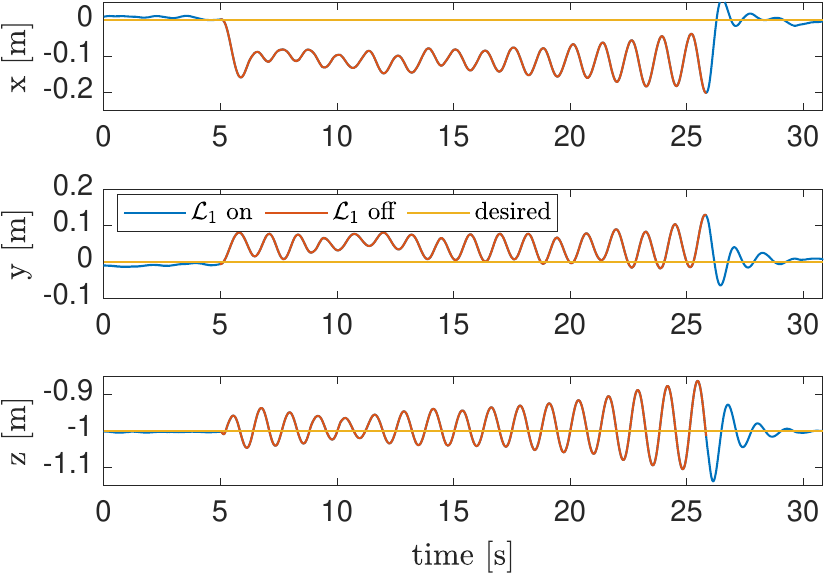}
    \caption{Tracking performance of the quadrotor with mixed propellers.}
    \label{fig:mixed prop tracking}
\end{figure}

\subsection{Voltage drop}
In this experiment, we test the influence of the battery's voltage drop on the tracking performance. Voltage drop is a common issue for quadrotors, as the ubiquitously used Lithium polymer (LiPo) battery's voltage drops from 4.2~V (fully charged) to 3.2 V (minimum voltage for healthy battery) for each cell when discharged. The varying voltage during the quadrotor's flight will affect the generated thrust on each rotor if not compensated. Specifically, the thrust reduces quadratically to the drop in battery voltage when the duty cycle command remains the same. We command the quadrotor to hover at 1 m while carrying a 0.5 kg payload. The extra payload is used to expedite the battery discharging process. For the altitude tracking, as shown in Fig~\ref{fig:voltage drop z tracking}, the altitude with $\mathcal{L}_1$ on stays steady on 1 m, whereas the altitude with $\mathcal{L}_1$ off drops as the voltage drops. It is worth noting that an initial altitude tracking error exists, which is caused by the extra payload and steady-state error of a geometric controller ($\mathcal{L}_1$ off). Figure~\ref{fig:voltage drop voltage} shows the voltage drop with $\mathcal{L}_1$ on and off. The voltage drop is almost identical in these two cases.
\begin{figure}[H]
      \begin{subfigure}[b]{\columnwidth}
         \centering
         \includegraphics[width=\columnwidth]{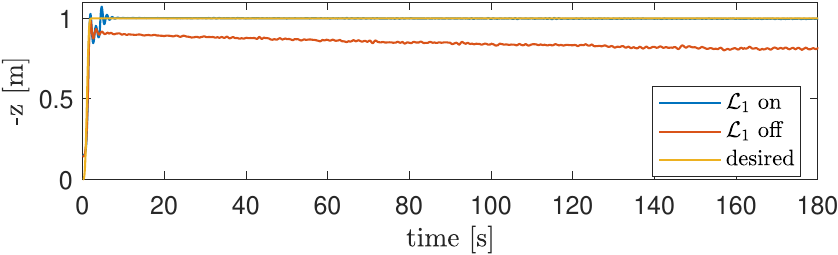}
         \caption{Altitude of the quadrotor. The altitude decreases with $\mathcal{L}_1$ off due to thrust loss from voltage drop.}
        \label{fig:voltage drop z tracking}
     \end{subfigure}
     \\
     \begin{subfigure}[b]{\columnwidth}
         \centering
         \includegraphics[width=\columnwidth]{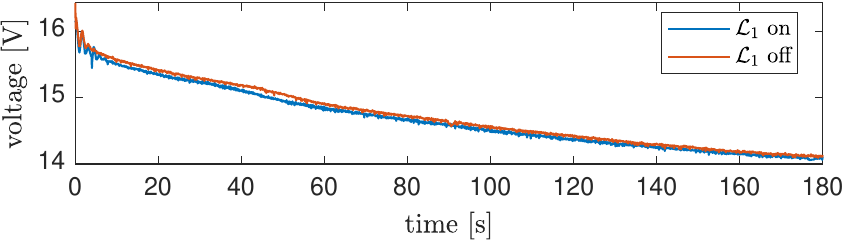}
         \caption{Battery voltage during hover.}
        \label{fig:voltage drop voltage}
     \end{subfigure}
     \caption{Results for the compensation of battery voltage drop by the $\mathcal{L}_1$ adaptive control.}
     \label{fig:voltage}
\end{figure}

\subsection{Ground effect}\label{subsec: ground effct}
The ground effect refers to the aerodynamic effect when the quadrotor flies close to the ground or a surface beneath it. In this experiment, we place a gym mat beneath the quadrotor's path. The top surface of the mat is at 0.95~m, whereas the altitude of the desired trajectory is 1 m (see Fig~\ref{fig:ground effect illustration} for the illustration of the setup). This mat induces the ground effect when the quadrotor flies on top of the mat: we name the area above the mat by ``ground effect zone.'' The quadrotor's altitude tracking is majorly impacted due to the reflected downwash when the quadrotor flies through the ground effect zone. Figure~\ref{fig:ground effect z tracking} shows the comparison of altitude tracking between $\mathcal{L}_1$ on and off, where the quadrotor enters and leaves the ground effect zone at 1.55 and 4.86~s, respectively.
\setcounter{figure}{11}
\begin{figure}
     \begin{subfigure}[b]{\columnwidth}
         \centering
         \includegraphics[width=\columnwidth]{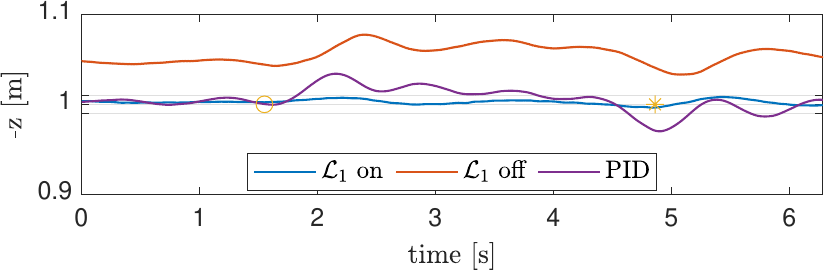}
         \caption{Altitude tracking in the ground effect zone. The desired altitude is 1 m. The circle and star indicate the time of entering and leaving the zone, respectively.}
        \label{fig:ground effect z tracking}
     \end{subfigure}
     \\
      \begin{subfigure}[b]{\columnwidth}
         \centering
         \includegraphics[width=\columnwidth]{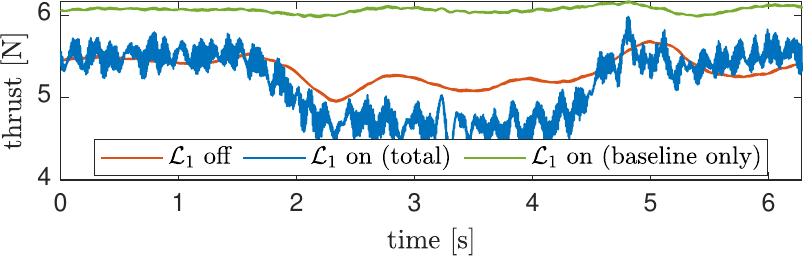}
         \caption{Thrust commands in the ground effect zone for $\mathcal{L}_1$ on and off.}
        \label{fig:ground effect thrust command}
     \end{subfigure}
     \caption{Results for the ground effect experiment}
     \label{fig:ground effect}
     \vspace{-3mm}
\end{figure}
\begin{figure*}
    \begin{subfigure}[b]{0.32\textwidth}
         \centering
         \includegraphics[width=\columnwidth]{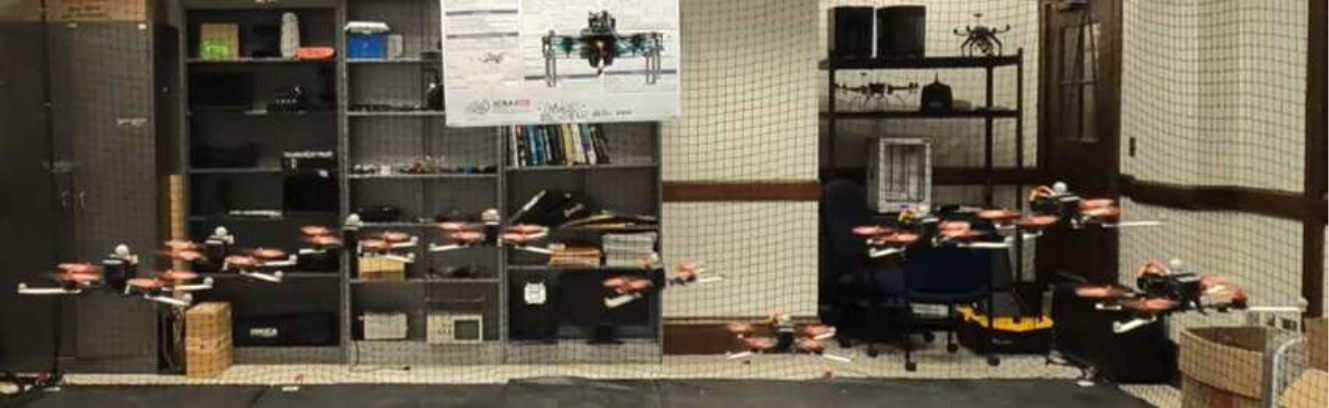}
         \caption{$\mathcal{L}_1$ off}
        \label{fig:downwash L1 off}
     \end{subfigure}
     \hfill
     \begin{subfigure}[b]{0.32\textwidth}
         \centering
         \includegraphics[width=\columnwidth]{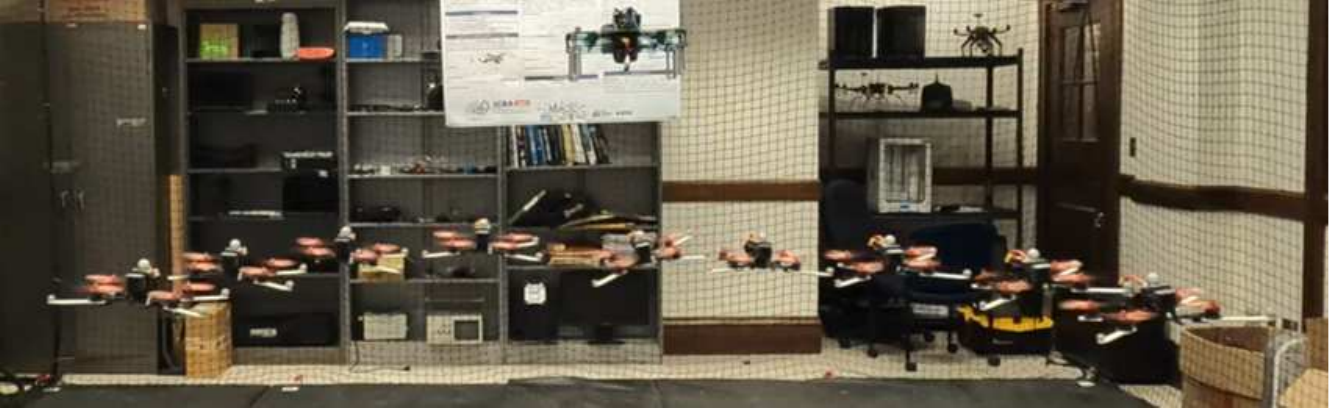}
         \caption{$\mathcal{L}_1$ on}
        \label{fig:downwash L1 on}
     \end{subfigure}
     \hfill
     \begin{subfigure}[b]{0.32\textwidth}
         \centering
         \includegraphics[width=\columnwidth]{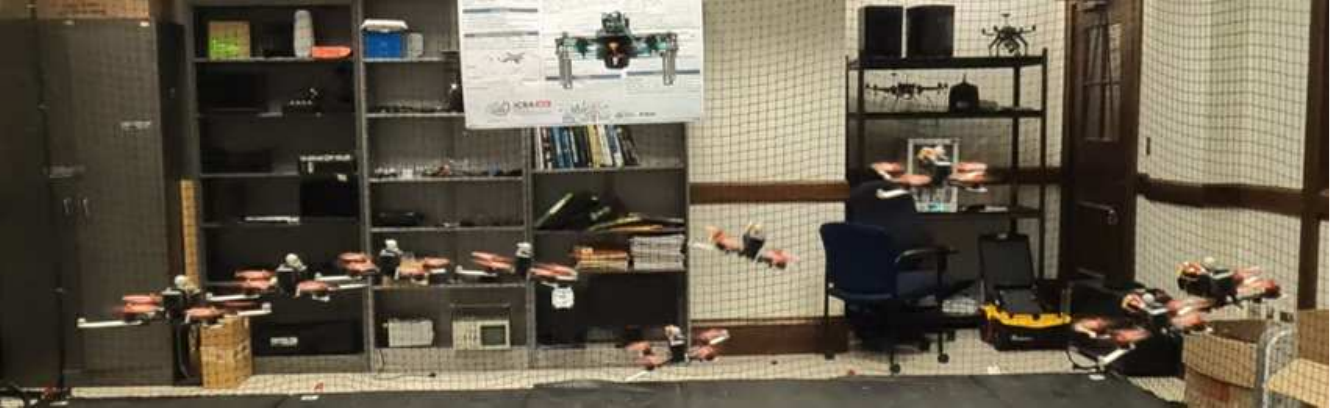}
         \caption{PID}
        \label{fig:downwash PID}
     \end{subfigure}
     \caption{Snapshots of the downwash experiment. The snapshots are taken from 0 s to 3.14 s corresponding to Figure~\ref{fig:downwash} with equal time intervals.}
     \label{fig:downwashsnapshot}
     \vspace{-5mm}
\end{figure*}

In the case of $\mathcal{L}_1$ on, the quadrotor closely tracks the desired altitude despite small perturbations when entering and leaving the ground effect region, with a maximum deviation within 1 cm.
Since the geometric controller ($\mathcal{L}_1$ off) admits a steady-state error and  results in the quadrotor flying at a higher altitude than the case with $\mathcal{L}_1$ on. This behavior leads to the quadrotor with $\mathcal{L}_1$ off suffering less from the ground effect than the one with $\mathcal{L}_1$ on. Despite this fact, the amplitude of perturbation with $\mathcal{L}_1$ on is kept smaller than in the case of $\mathcal{L}_1$ off.

We show the thrust commands in the cases of $\mathcal{L}_1$ on and off in Fig~\ref{fig:ground effect thrust command}. It can be seen that the total thrust command (baseline plus $\mathcal{L}_1$ adaptive control) is reduced in the ground effect zone when $\mathcal{L}_1$ is on, which compensates for the reflected downwash. Additionally, the $\mathcal{L}_1$ adaptive control applies a constant offset to compensate for the extra thrust (shown in green in Fig~\ref{fig:ground effect thrust command}) commanded by the geometric controller (which results in the steady-state error if not compensated for).
The thrust command in the case of $\mathcal{L}_1$ off is also reduced when the quadrotor flies through the ground effect zone. However, the reason for the thrust reduction with $\mathcal{L}_1$ off is fundamentally different from the one discussed above when $\mathcal{L}_1$ is on. 
The quadrotor was lifted by the ground effect, which causes a bigger altitude tracking error in the ground effect zone, as shown in Fig~\ref{fig:ground effect z tracking}. The lifted altitude leads to reductions in the commanded thrusts as the baseline controller reduces the thrusts to lower the altitude for smaller tracking errors. 
\subsection{Downwash}\label{subsec: downwash}
In this experiment, we test the trajectory tracking performance subject to downwash. Downwash is the complex disturbance generated by one rotorcraft on top of another vehicle, which severely degrades the attitude and position tracking of the latter. It is generally difficult to establish an accurate model based on the relative motion between two vehicles. Hence, adaptive control methods are preferred. We command a heavier quadrotor (1.4~kg) to hover at 1.6~m altitude, right on top of one point on the trajectory of the tested quadrotor, as shown in Fig~\ref{fig:downwashsnapshot}. The comparison of the altitude tracking performance between $\mathcal{L}_1$ on, $\mathcal{L}_1$ off, and PID is shown in Fig~\ref{fig:downwash}, with the downwash starting to affect the tested quadrotor around 1.5~s. The altitude tracking, displayed in Fig~\ref{fig:downwash}, indicates that the $\mathcal{L}_1$ adaptive control can quickly capture the effect of downwash and compensate for it, resulting in a fast transient response, whereas the quadrotor without compensation ($\mathcal{L}_1$ off) experienced instant deviation from the desired altitude by around 0.3 m. The PID controller performs poorly with a similar amount of deviation as in the case with $\mathcal{L}_1$ off. 
\setcounter{figure}{13}
\begin{figure}[ht]
\vspace{-1mm}
    \centering   \includegraphics[width=\columnwidth]{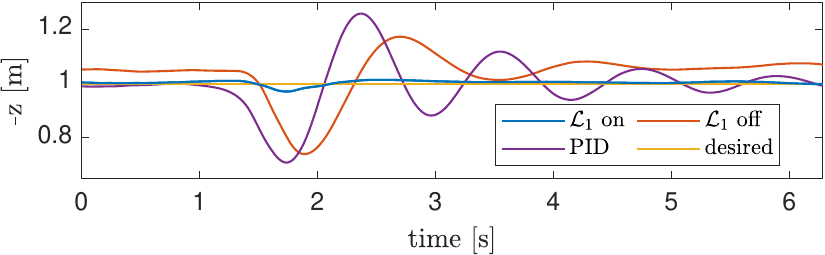}
     \caption{Altitude tracking in the experiments with downwash. The downwash starts to influence the quadrotor around 1.5 s. The heavier quadrotor that generates the downwash drifted around the desired hover position by 0.1~m, which leads to the different start time of downwash of the tested quadrotor.}
     \label{fig:downwash}\vspace{-3mm}
\end{figure}

\subsection{Hanging off-center weights}
In this experiment, we hang weights to the quadrotor while hovering. 
Specifically, the weights are hung right beneath the front-left motor, which is off the quadrotor's center of gravity. 
This experiment is a simplified scenario of delivery drones that pick up a package, albeit at an off-center location. 
The weights will introduce an unknown force and moment to the quadrotor in this setting, adding the complexity of stabilizing the quadrotor. 
We hang weights of 100, 200, and 300 g in the experiments. The tracking errors in these three cases are shown in Fig~\ref{fig: hang weight}. We compare three controllers: $\mathcal{L}_1$ off, $\mathcal{L}_1$ on, and PID. As the weight increases, the tracking error increases for all cases. $\mathcal{L}_1$ off shows the worst tracking performance and eventually crashes the quadrotor once the 300 g weight is hung. $\mathcal{L}_1$ on has advantages over PID in terms of faster transient and smaller steady-state error, as shown in Table~\ref{tb: hang weight comparison}, where the advantage is clearer as the weights increase. We compute the steady-state error as the average tracking error between 10 and 12 seconds. The settling time is computed as the time duration from when the weight is hung until the quadrotor begins to maintain a tracking error less than the steady-state error plus 0.015 m.
\begin{table}[H]
\vspace{-4mm}
	\centering
	\small
	\caption{Performance comparison of the hanging off-center weights experiments.}
	\begin{tabular}{c|ccc | c c c}
		\toprule[1pt]
		& \multicolumn{3}{c|}{Steady-state Error [m]} &\multicolumn{3}{c}{Settling Time [s]} \\  \hline
        Weights & \ellone on& \ellone off & PID & \ellone on& \ellone off & PID\\
        100g & 0.004 & 0.044 & \textbf{0.003} &\textbf{0.77}&1.18&1.56\\
        200g & \textbf{0.004} & 0.125 & 0.006 &\textbf{2.77}&3.81&3.84\\
        300g & \textbf{0.007} & Crash & 0.02 &\textbf{3.36}&Crash&6.6\\
        
		\bottomrule[1pt]
	\end{tabular}\label{tb: hang weight comparison}
    \vspace{-5mm}
\end{table}
\begin{figure}
    \centering
    \includegraphics[width = \columnwidth]{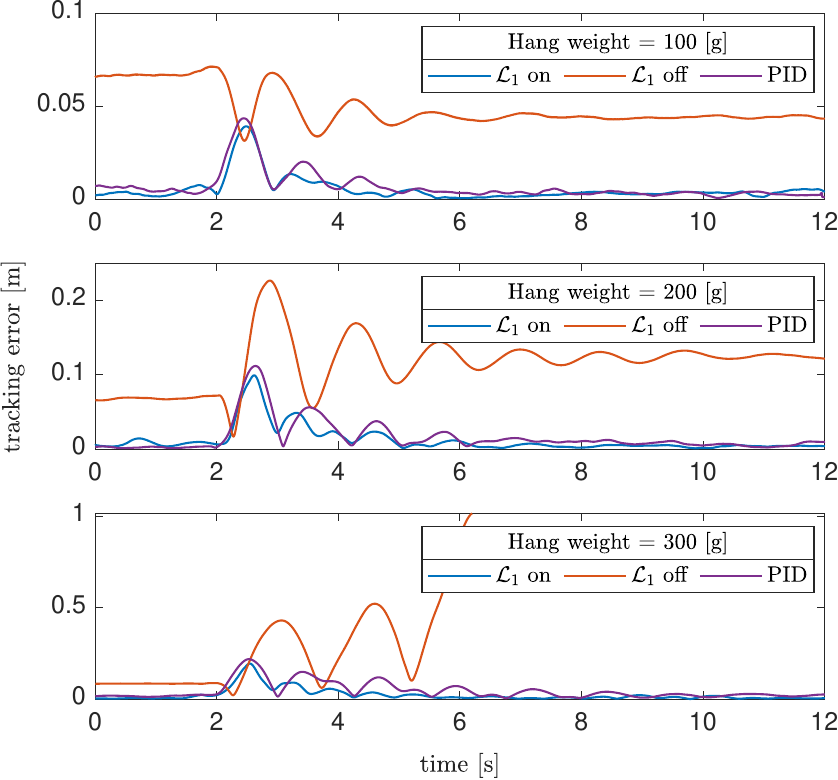} \caption{Tracking error of the quadrotor when hanging off-center weights in it starting at 2 s.}
    \label{fig: hang weight}
    \vspace{-5mm}
\end{figure}

\setcounter{figure}{15}
\begin{figure*}
\begin{subfigure}{\textwidth}
    \includegraphics[width = \textwidth]{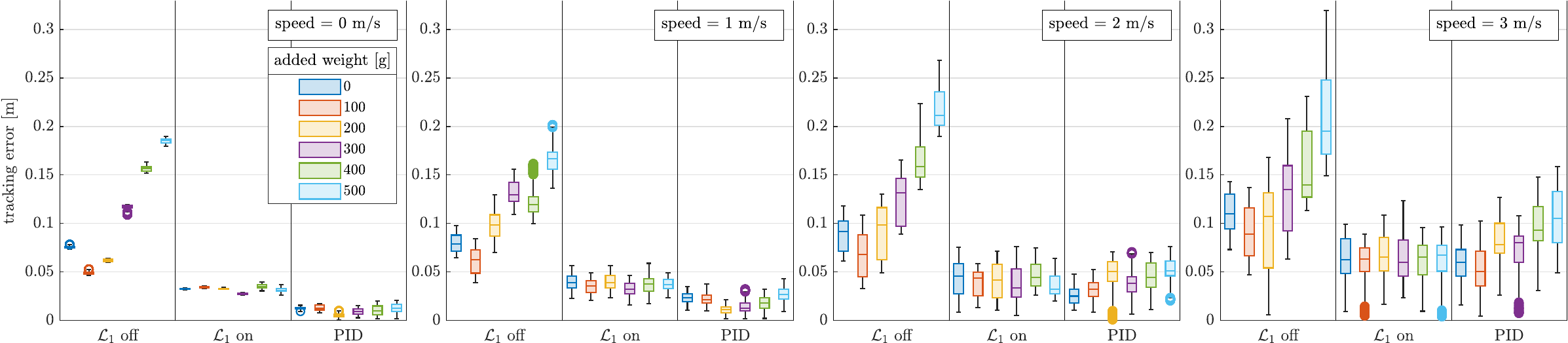}
        \caption{Added weights}
        \label{fig: figure 8 added Weight BenchMark}
\end{subfigure}
\begin{subfigure}[b]{\textwidth}
        \includegraphics[width = \textwidth]{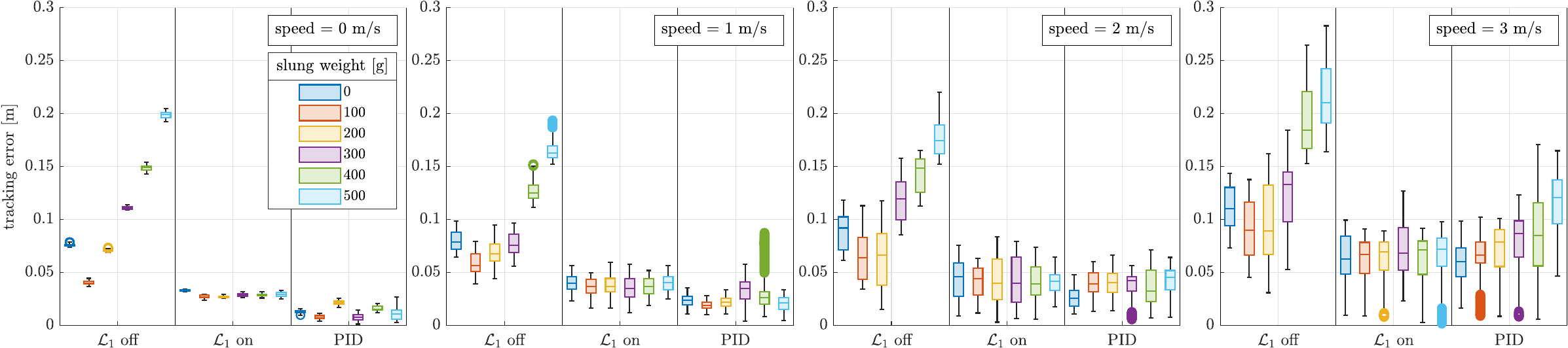}
        \caption{Slung weights}
        \label{fig: figure 8 slung Weight BenchMark}
    \end{subfigure} 
    \caption{Tracking errors of the quadrotor when following a figure 8 trajectory.}
    \label{fig: figure 8 benchmark}
\end{figure*}

\subsection{Benchmark experiments with added weights and slung weights}
We would like to understand the systematic performance improvement by the $\mathcal{L}_1$ adaptive control in addition to the case-by-case scenarios in the experimental results shown above. For the systematic experiment, we attach weights to the bottom of the quadrotor, where the weights range from 100 to 500~g with 100~g increments. 
We refer to this case as ``added weights.'' We also include the ``slung weights,'' where the weights are hung by a cord with 0.4 m~length.
For both types of weights, we command the quadrotor to fly two trajectories: 1) a 1~m radius circular trajectory with linear speed ranging from 0 (hover) to 2.5 m/s with a 0.5 m/s increment and 2) a figure 8 trajectory\footnote{The figure 8 trajectory is composed by sinusoidal functions $x(t) = 2\sin(vt/2.51)$, $y(t) = 1.5\sin(2vt/2.51)$, and $z(t) = 0.2\sin(vt/2.51)-1$ for $v$ being the maximum speed of the figure 8 trajectory at its center.} with the maximum speed ranging from 0 (hover) to 3 m/s with a 1 m/s increment. The comparisons on the tracking error among $\mathcal{L}_1$ off, $\mathcal{L}_1$ on, and PID
are shown in Fig~\ref{fig: figure 8 benchmark} (figure 8 trajectory) and Fig~\ref{fig: circular benchmark} (circular trajectory). For each case, data within 5 s (4000 data points in total) are used to make the box plots.

On the one hand, the tracking errors with $\mathcal{L}_1$ on are significantly smaller than with $\mathcal{L}_1$ off with the same trajectory, speed, and type of weights. On the other hand, with the same trajectory and speed, the tracking error with $\mathcal{L}_1$ on stays consistent despite the gradually increased weights. This consistency does not show up when $\mathcal{L}_1$ is off, since the tracking errors grow as the weights increase. 
PID performs considerably better than $\mathcal{L}_1$ off, with similar (and sometimes smaller) errors to that of $\mathcal{L}_1$ off. However, in extreme cases with relatively large weights and high speeds, PID cannot stabilize the quadrotor and eventually crashes the vehicle, as indicated by the red crosses in Fig~\ref{fig: figure 8 slung Weight BenchMark}. This happens again due to the incapability of PID in handling fast-changing uncertainties: in this case, the uncertainties come from the swing of the slung weights during the figure-8 maneuver. 
The RMSE comparisons are provided in Appendix~\ref{apdx: benchmark data}.
\setcounter{figure}{16}
\begin{figure*}[!h]
    \begin{subfigure}[b]{\textwidth}
        \includegraphics[width = \textwidth]{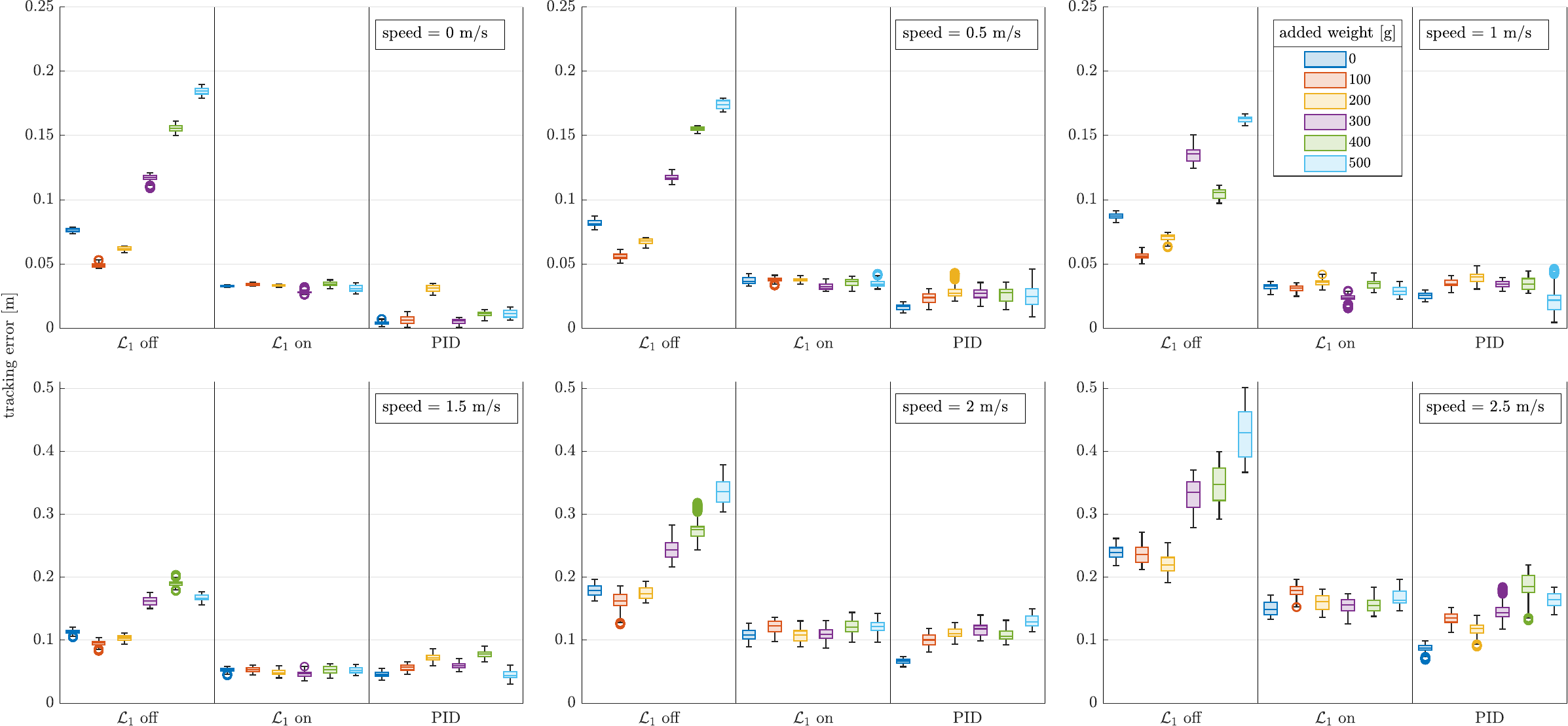}
        \caption{Added weights}
        \label{fig: circular added Weight BenchMark}
    \end{subfigure} 
    \\ 
    \begin{subfigure}[b]{\textwidth}
        \includegraphics[width = \textwidth]{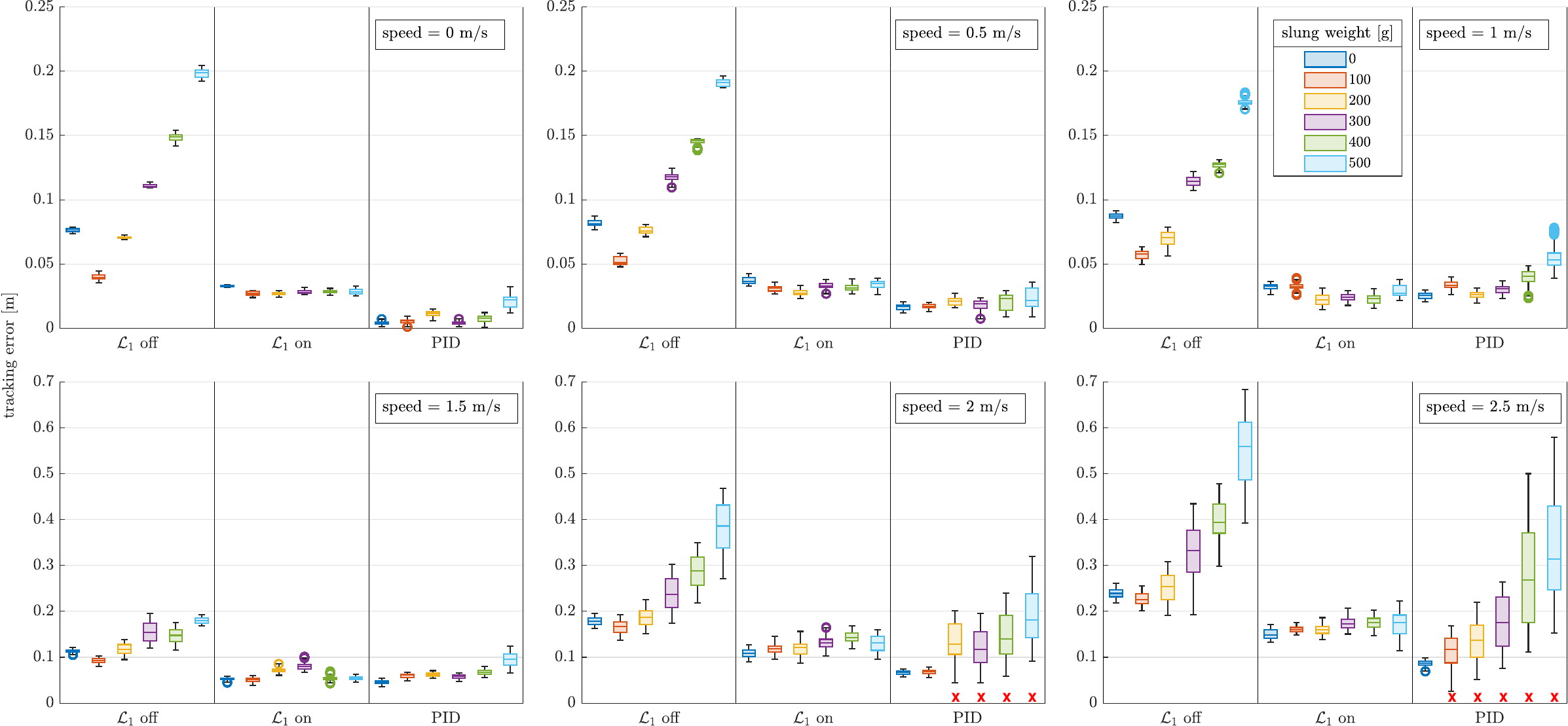}
        \caption{Slung weights}
        \label{fig: circular slung Weight BenchMark}
    \end{subfigure} 
    \caption{Tracking errors of the quadrotor when following a circular trajectory. The red crosses indicate that the trials crashed.}
    \label{fig: circular benchmark}
\end{figure*}
\subsection{Discussion}
We have validated the performance of the $\mathcal{L}_1$Quad in eleven types of uncertainties across various trajectories. 
In all experiments, the $\mathcal{L}_1$Quad  consistently achieves small tracking errors, which validates the uniform tracking error bound proved in Theorem~\ref{thm: main theorem}. By comparing $\mathcal{L}_1$Quad with other adaptive control approaches in Section~\ref{subsec: injected unc}, we demonstrate $\mathcal{L}_1$Quad's ability to handle uncertainties without the need for explicitly modeling them. 
Additionally, in the experiments where  $\mathcal{L}_1$ is turned on mid-flight (chipped propeller and mixed propeller), or when uncertainties suddenly impact the system (ground effect, downwash, hanging off-center weight, and tunnel), $\mathcal{L}_1$Quad quickly suppresses the impact and maintains a stable flight, which validates the unique advantage of $\mathcal{L}_1$Quad in providing transient guarantees in the presence of uncertainties. 
In harsh conditions (chipped propeller and hanging off-center weights), due to severe uncertainties, the quadrotor \textit{crashed} when flying without the $\mathcal{L}_1$Quad, as a contrast to stable flights with small tracking errors when flying with the $\mathcal{L}_1$Quad. These comparisons showcase our capability for handling uncertainties in harsh conditions.
Lastly, only \textit{one} set of controller parameters is used in $\mathcal{L}_1$Quad through all the experiments, highlighting its architectural benefit of reduced tuning efforts compared with the baseline PID approach that requires case-by-case tuning.

\section{Conclusion}\label{sec:conclusions}
In this paper, we present $\mathcal{L}_1$Quad, an \ellone adaptive control design for quadrotors. Our design
augments a geometric controller with an $\mathcal{L}_1$ adaptive controller for both rotational and translational dynamics. We lump the uncertainties as unknown forces and moments, which can be quickly estimated and compensated for by the \ellone adaptive control. No case-by-case modeling of the uncertainties is required, which makes the architecture applicable to uncertainties in a wide spectrum of applications.

Theoretically, we prove that the piecewise-constant adaptation law can estimate uncertainties with the error proportional to the sampling time, which can be made arbitrarily small subject only to hardware limitations. With the fast and accurate uncertainty estimation, the $\mathcal{L}_1$Quad enables computable uniform and uniform ultimate bounds, which guarantees the tracking performance of the quadrotor. Furthermore, the width of the ultimate bound can be adjusted by tuning the filter bandwidth of the low-pass filter and the sampling time.

Experimentally, we validate our approach on a custom-built quadrotor through extensive experiments with various types of uncertainties. The superior performance of the $\mathcal{L}_1$Quad is illustrated by its consistently smaller tracking error when compared with other controllers in all the experiments, where only \textit{one} set of controller parameters is used in the $\mathcal{L}_1$Quad without tuning case by case.

As for the limitations, the $\mathcal{L}_1$Quad only compensates for matched uncertainties. Active compensation for unmatched uncertainties using the unmatched uncertainty estimates requires further investigation. Additionally, the theoretical results are developed based on the assumption that the states of the quadrotor can be accurately measured. Ongoing investigation aims to address this issue following architectures of output feedback in the $\mathcal{L}_1$ adaptive control theory.

Future work includes incorporating the current bounds in the safe feedback motion planning architecture~\cite{lakshmanan2020safe} and applying learning techniques (e.g., \cite{gahlawat2021contraction}) to better characterize the uncertainties using data. Learned information could be incorporated into the proposed control framework, which alleviates the workload for the $\mathcal{L}_1$ adaptive controller and achieves better performance with improved robustness.

\bibliographystyle{IEEEtran}
\bibliography{ref}
\appendices
\section{Proof of Proposition~\ref{prop:geometricproposition}}\label{apdx: proof of geometric}
The original proof in~\cite[Appendix D, Section (c)]{lee2010arxiv} has a simplifying assumption that restricts the analysis to the domain defined in~\cite[Eq. (74)]{lee2010arxiv}. We rewrite the proof in~\cite[Appendix D Part (c) and Part (d)]{lee2010arxiv} (starting from~\cite[Eq. (79)]{lee2010arxiv}) and avoid this restriction. 

\begin{proof}
Consider the translational error dynamics~\cite[Eq. (75)]{lee2010arxiv} 
\[
m\dot{e}_v = mge_3 - fRe_3 - m\Ddot{x}_d,
\]
and its Lyapunov function candidate:
\begin{equation}\label{lyapunovtranslation}
    V_1 = \frac{1}{2}K_p \norm{e_p}^2 + \frac{1}{2}m \norm{e_v}^2 + c_1 e_p\cdot e_v,
\end{equation}
with $c_1$ being a positive constant.
The Lie derivative of $V_1$ is 
\begin{align*}
\dot{V}_1 =& K_p e_p \cdot e_v + e_v \cdot \left(  -K_p e_p - K_v e_v - X \right) + c_1e_v \cdot e_v \\ &+ \frac{c_1}{m}e_p \cdot \left(  -K_p e_p - K_v e_v - X \right),
\end{align*}
where we applied the equality $m \dot{e}_v=-K_p e_p -K_v e_v -X$~\cite[Eq. (78)]{lee2010arxiv} with 
\[
X=\frac{f}{e_3^{\top} R_d^{\top} R e_3} \left(  (e_3^{\top} R_d^{\top} R e_3)Re_3 -R_de_3 \right).
\]
Rearranging the equation above, we obtain
\begin{align*}
    \dot{V}_1 =& -(K_v -c_1) \norm{e_v}^2 -\frac{c_1 K_p}{m}\norm{e_p}^2 \\ & - \frac{c_1 K_v}{m}e_p \cdot e_v+ X \cdot \left(  -\frac{c_1}{m}e_p - e_v \right).
\end{align*}
Since $X$ is bounded such that $\norm{X} \leq \left(  K_p \norm{e_p} + K_v \norm{e_v} + H \right) \norm{e_R}$~\cite[Eq. (81)]{lee2010arxiv}, and $H > \norm{-mg + m \Ddot{x}_d}$~\cite[Eq. (25)]{lee2010arxiv}, using Cauchy-Schwarz inequality, we have
\begin{align}
    \dot{V}_1 \leq&  \left(  K_p \norm{e_p} + K_v \norm{e_v} + H \right) \norm{e_R} \left(  \frac{c_1}{m}\norm{e_p} + \norm{e_v} \right)\nonumber \\ &-(K_v -c_1) \norm{e_v}^2 -\frac{c_1 K_p}{m}\norm{e_p}^2 \nonumber \\ &- \frac{c_1 K_v}{m} e_p \cdot e_v .\label{equ: bound on v1 dot}
\end{align}
Our proof will deviate from the original proof now. Define $\alpha > 0$ such that $ \norm{e_R} \leq \alpha < 1$~\cite[Before eq. (81)]{lee2010arxiv}. Substituting $\alpha$ in the inequality above, we have
\begin{align}\label{equ: v1dotfinal}
    \dot{V}_1 \leq& -(K_v(-\alpha + 1) -c_1 ) \norm{e_v}^2 -\frac{c_1 K_p}{m}(1-\alpha)\norm{e_p}^2 \nonumber\\&+  \left(\frac{c_1 K_v}{m}(1+\alpha) +K_p\alpha\right)\norm{e_p}\norm{e_v} \nonumber\\ &+ \norm{e_R} \left(  \frac{c_1}{m}H\norm{e_p} + H\norm{e_v} \right).
\end{align}
Consider the rotational error dynamics~\cite[Eq. (56)]{lee2010arxiv} 
\[
J\dot{e}_\Omega = J\dot{\Omega} + J\left(\hat{\Omega} R^\top R_d \Omega_d - R^\top R_d \dot{\Omega}_d \right),
\]
and its Lyapunov function candidate:
\[
V_2 = \frac{1}{2}e_\Omega \cdot Je_\Omega + K_R \Psi (R,R_d) + c_2 e_R \cdot e_\Omega,
\]
with a positive constant $c_2$. The Lyapunov function candidate for the complete dynamics is $V = V_1 + V_2$, where $V$ is defined in~\eqref{equ:lyapunovforcomplete}. Using the bound $\frac{1}{2} \norm{e_R}^2\leq \Psi(R,R_d)\leq \frac{1}{2-\psi_1}\norm{e_R}^2$ in~\cite[Eq. (73)]{lee2010arxiv}, the Lyapunov function candidate is bounded by
\begin{equation}\label{equ:lyapunovbound}
    z_1^{\top}M_{11}z_1 + z_2^{\top}M_{21}z_2 \leq V \leq z_1^\top M_{12}z_1 + z_2^{\top}M_{22}z_2,
\end{equation}
where $z_1=[\norm{e_p},\norm{e_v}]^{\top}$, $z_2=[\norm{e_R},\norm{e_\Omega}]^{\top}$, and the matrices $ M_{11}$, $M_{12}$, $M_{21}$, and $M_{22}$ are defined in Section~\ref{sec: performance analysis}.
Furthermore,
\begin{align}\label{equ:lyapunovbound2}
    &\lambda_m(M_{11})\norm{z_1}^2 + \lambda_m(M_{21})\norm{z_2}^2 \leq V  \nonumber\\&\leq \lambda_M(M_{12})\norm{z_1}^2 + \lambda_M(M_{22})\norm{z_2}^2.
\end{align}
Next, we derive the bound on $\dot{V} = \dot{V}_1 + \dot{V}_2$. The bound on $\dot{V}_1$ is shown in~\eqref{equ: v1dotfinal}. For $\dot{V}_2$, we follow the derivation in~\cite[(58)]{lee2010arxiv}
\begin{align}
\dot{V}_2 =& -K_\Omega\norm{e_\Omega}^2-c_2K_Re_R \cdot J^{-1}e_R + c_2C(R_d^\top R)e_\Omega \cdot e_\Omega \nonumber \\ & -c_2 K_\Omega e_R \cdot J^{-1}e_\Omega, 
\end{align}
where $C(R_d^\top R) = \left( \text{tr}(R^\top R_d)I-R^\top R_d \right)/2$, and $\norm{C(R_d^\top R)}_2 \leq 1$~\cite[Eq. (54)]{lee2010arxiv}. The derivative $\dot{V}$ is bounded by
\begin{equation}\label{lyapunovcondition}
    \dot{V} \leq -z_1^{\top} W_1 z_1 + z_1^{\top} W_{12} z_2 -z_2^{\top} W_2 z_2,
\end{equation}
where $W_1$, $W_{12}$, $W_2$ are defined in Section~\ref{sec: performance analysis}.
We choose positive constants $K_p$, $K_v$, $K_{\Omega}$, $K_R$, $c_1$, and $c_2$ such that $M_{11}$, $M_{21}$, $M_{12}$, $M_{22}$, $W_1$, and $W_2$ are positive definite matrices, and $W \succ 0$ where $W$ is defined in~\eqref{equ:definition of W}.
Therefore, $\dot{V}$ is negative-definite, and the zero equilibrium of the closed-loop tracking errors $(\norm{e_p},\norm{e_v},\norm{e_R},\norm{e_\omega})$ is exponentially stable. Specifically, we have
\begin{equation}
    \dot{V} \leq -\beta V,
\end{equation}
where $\beta$ is defined in~\eqref{equ: beta condition}.
Also, using~\eqref{equ: tube definition} and~\eqref{equ:lyapunovbound}, we have
\begin{equation}
    \underline{\gamma}d(x,x_d)^2 \leq V \leq \overline{\gamma}d(x,x_d)^2,
\end{equation}
where $\underline{\gamma}$ and $\overline{\gamma}$ are defined in~\eqref{equ:definition of gamma lower bar} and~\eqref{equ:definition of gamma upper bar}, respectively.
Without loss of generality, we consider $K_p$, $K_v$, $K_R$, $K_\Omega$ to be scalars for simplicity. However, the extension to the diagonal 3-by-3 matrices is straightforward, where maximum or minimum eigenvalues of these matrices will be used.
\end{proof}

\section{Proof of Proposition~\ref{prop:bounds of Ps Ub and Uad}}\label{apdx: proof of prop 3}
\begin{proof}
Since $x(t) \in \mathcal{O}(x_d(t),\rho)$, we have $\norm{e_p(t)}$, $\norm{e_v(t)},$ $\norm{e_\Omega(t)}$, and $\norm{e_R(t)}$ upper bounded by $\rho$ (following the definition of $\mathcal{O}(x_d(t),\rho)$ in~\eqref{equ: tube definition}). Considering that $c_1$ and $m$ are constants and $\norm{Re_3}=1$ by the definition of the rotation matrix, we can conclude that all three bounds in~\eqref{bounds of P1 P2} hold. Furthermore, following the control law in~\eqref{eq:thrust control} and~\eqref{eq:torque control}, we conclude that $\norm{u_b(t)}$ is also bounded. Define $\Delta_{u_b}$ such that $\norm{u_b(t)} \leq \Delta_{u_b}$, for all $t\geq 0$. 

Since $C(s)=\omega/\left(s+\omega\right)$ with DC gain equal to 1, we have $\norm{C(s)}_{\mathcal{L}_1}=1$. Following~\cite[Lemma A.7.1]{hovakimyan2010L1}, we obtain 
$
\norm{u_{ad}(t)} \leq  \norm{C(s)}_{\mathcal{L}_1}\norm{\hat{\sigma}(t)}_{\mathcal{L}_\infty} \nonumber 
$. Therefore, in order to prove that $\norm{u_{ad}(t)}$ is bounded, we need to show $\norm{\hat{\sigma}(iT_s)}$ is bounded for all $i \in \mathbb{N}$. Using the initial condition $\tilde{z}(0)=0$ of the prediction error dynamics in~\eqref{equ:prediction error dynamics}, we conclude that for $i=0$, $\hat{\sigma}(iT_s)=0$ from the definition of $\hat{\sigma}(iT_s)$ introduced in~\eqref{equ: sigma hat}. For the case of $i \geq 1$, recall the representation of the prediction error from~\eqref{equ: prediction error final}
\begin{equation*}
\tilde{z} (iT_s) =-\textstyle\int_{(i-1)T_s}^{iT_s} \exp(A_s(iT_s-\tau))  \Bar{B}(\tau)\sigma(\tau)\dd\tau.
\end{equation*}
Since $\bar{B}(t)$ is continuous by design and $\sigma(t)$ is continuous by Assumption~\ref{aspt: main assumption} in the main manuscript, considering that every diagonal element in $\exp(A_s(iT_s-\tau))$ is positive for $i\geq 1$ and $\tau \in [(i-1)T_s, iT_s)$, we apply the mean value theorem~\cite{comenetz2002calculus} in an element-wise manner to the right-hand side of the above equation. As a result, there exists $(\bar{B}^{\ast}\sigma^{\ast})_k \triangleq (\bar{B}(\tau_k)\sigma(\tau_k))_k$, for $k=1,\dots,6$ and $\tau_1,\dots,\tau_6 \in [(i-1)T_s,iT_s)$, such that
\begin{align}\label{equ: mean value theorem}
\tilde{z}(iT_s) \nonumber 
=&-\textstyle\int_{(i-1)T_s}^{iT_s} \exp(A_s(iT_s-\tau)) \dd\tau \Bar{B}^\ast\sigma^\ast \nonumber \\ =&A_s^{-1}(I - \exp(A_sT_s))\Bar{B}^\ast \sigma^\ast.
\end{align}

Substituting $\tilde{z}(iT_s)$ from~\eqref{equ: mean value theorem} to~\eqref{equ: sigma hat}, we have
\begin{equation}\label{equ: estimated uncertainty i geq 1}
\hat{\sigma}(iT_s)=\bar{B}(iT_s)^{-1}\exp(A_sT_s)\bar{B}^\ast\sigma^\ast,
\end{equation}
for $i \geq 1$. Using the bounds of $\norm{\bar{B}}$, $\norm{\bar{B}^{-1}}$ in~\eqref{bounds on Bs}, and the bound of $\sigma(t,x(t))$ from~\eqref{equ: bounds of sigmas}, we have 
\begin{equation}\label{bound of sigma hat new}
    \norm{\hat{\sigma}(iT_s)} \leq \Delta_{\hat{\sigma}},
\end{equation}
where $\Delta_{\hat{\sigma}} \triangleq \Delta_{\bar{B}^{-1}} \Delta_{\bar{B}} \Delta_{\sigma}$ for all $i \in \mathbb{N}$. 

Finally, we conclude that
$\norm{u_{ad}(t)} \leq \Delta_{\hat{\sigma}},$ for all $t \geq 0$.
\end{proof}

\section{Proof of Proposition~\ref{prop: estimation error bound}}\label{apdx: proof of prop 4}
Before proving Proposition~\ref{prop: estimation error bound}, we first introduce the following proposition.
\setcounter{proposition}{4}
\begin{proposition}\label{proposition: bound on x dot}
If the state $x(t)$ of the uncertain system in~\eqref{equ:uncertain system with full state and with L1} satisfies $x(t) \in \mathcal{O}(x_d(t),\rho)$ for all $t \in [0,\tau]$ with some $\tau > 0$, then the following inequality holds
\begin{equation}
    \norm{\dot{x}}_{\mathcal{L}_\infty}^{[0,\tau]} \leq \Delta_f+ \Delta_{\bar{B}}\Delta_\sigma+\Delta_{B_F} (\Delta_{u_b}+\Delta_{\hat{\sigma}}).
\end{equation}
\end{proposition}

\begin{proof}
Using the dynamics of the uncertain system in~\eqref{equ:uncertain system with full state and with L1}, we obtain
\begin{align}
\norm{\dot{x}(t)} \leq & \norm{f(x(t))}+ \norm{\bar{B}_F(R(t))}\norm{\sigma(t,x(t))} \nonumber\\ &+\norm{B_F(R(t))}\norm{  u_b(t) +u_{ad}(t)}.
\end{align}
With the help of the bounds in Assumption~\ref{aspt: main assumption},~\eqref{bounds on Bs},~\eqref{bounds of P1 P2}, and~\eqref{bounds on ub and sigmahat}, we can conclude that
\begin{equation*}
\norm{\dot{x}}_{\mathcal{L}_\infty}^{[0,\tau]} \leq \Delta_f+ \Delta_{\bar{B}}\Delta_\sigma+\Delta_{B_F} (\Delta_{u_b}+\Delta_{\hat{\sigma}})  = \phi_1,
\end{equation*}
where $\phi_1$ is defined in~\eqref{def: phi1}. 
\end{proof}

We now move to prove Proposition~\ref{prop: estimation error bound}.
\begin{proof}
We first derive the bound on $\norm{\hat{\sigma}(t,x(t)) - \sigma(t,x(t))}$ for any $t \in [ iT_s,(i+1)T_s )$, and $i \geq 1$. The bound on $\norm{\hat{\sigma}(t,x(t)) - \sigma(t,x(t))}$ for $t \in [0,T_s)$ will be discussed at the end of this section. 
Subtracting the uncertainty $\sigma(t)$ with the estimated uncertainty from~\eqref{equ: estimated uncertainty i geq 1}, we have
\begin{align*}
&\norm{\sigma(t)-\hat{\sigma}(t)} \nonumber 
\\
&= \norm{\sigma(t) - \Bar{B}(iT_s)^{-1} \exp(A_sT_s) \Bar{B}^\ast\sigma^\ast } \nonumber
\\
&= \norm{\Bar{B}(iT_s)^{-1}(\Bar{B}(iT_s))\sigma(t)-\exp(A_sT_s)\Bar{B}^\ast\sigma^\ast)} \nonumber 
\\
&\leq \norm{\Bar{B}(iT_s)^{-1}}\norm{\Bar{B}(iT_s))\sigma(t)-\exp(A_sT_s)\Bar{B}^\ast\sigma^\ast}.
\end{align*}
Since $\norm{q} \leq \sqrt{n} \max_{k \in \{1,\dots,n\} }\left|(q)_k\right| \leq \sqrt{n}\norm{p}$ for any $q, p \in \mathbb{R}^n$, as long as $(p)_k=(q)_k$, we rewrite the inequality above:
\begin{align}\label{eq: long equations 1}
&\norm{\sigma(t)-\hat{\sigma}(t)} \nonumber 
\\
&\leq \sqrt{6}\norm{\Bar{B}(iT_s)^{-1}} \left|( \Bar{B}(iT_s))\sigma(t)-\exp(A_sT_s)\Bar{B}^\ast\sigma^\ast)_j\right| \nonumber
\\
&\leq \sqrt{6}\left\| \Bar{B}(iT_s))\sigma(t) \right.  \left.-\exp(A_sT_s)\Bar{B}(\tau_j)\sigma(\tau_j)\right\| \nonumber \\
&\hspace{4mm} \times \norm{\Bar{B}(iT_s)^{-1}} ,
\end{align}
where $j\triangleq \arg \max _{j \in \{1,\dots,6\} } \Phi_0$, 
and 
$\Phi_0 = \left|( \Bar{B}(iT_s)\sigma(t)-\exp(A_sT_s)\Bar{B}^\ast\sigma^\ast)_j\right|.$

Adding and subtracting $\bar{B}(iT_s)\sigma(\tau_j)$ to the first norm on the right-hand side of the inequality above, we have
\begin{align}\label{estimation error bound derivation}
&\norm{\sigma(t) - \hat{\sigma}(t)} \nonumber 
\\ 
&\leq \sqrt{6}\left( \norm{\Bar{B}(iT_s)\sigma(\tau_j) -\exp(A_sT_s)\Bar{B}(\tau_j)\sigma(\tau_j)}\right. \nonumber 
\\ 
&\left.\hspace{11mm}+\norm{\Bar{B}(iT_s)}\norm{\sigma(t)-\sigma(\tau_j)}\right) \norm{\Bar{B}(iT_s)^{-1}}\nonumber 
\\
&\leq  
\sqrt{6}\norm{\Bar{B}(iT_s)^{-1}} \norm{\Bar{B}(iT_s)}\norm{\sigma(t)-\sigma(\tau_j)}
\nonumber 
\\
&\hspace{5mm}+ \left(  \norm{\Bar{B}(iT_s)-\bar{B}(\tau_j)} +\norm{(I-\exp(A_sT_s))\bar{B}(\tau_j)} \right) \nonumber  
\\ 
&\hspace{5mm} \times \sqrt{6}\norm{\sigma(\tau_j)}\norm{\Bar{B}(iT_s)^{-1}}\nonumber \\
&= \Phi_1(t) + \Phi_2(t), \vspace{-2mm}
\end{align}
where \vspace{-2mm}
\begin{equation*}
\Phi_1(t)  \triangleq \sqrt{6}\norm{\Bar{B}(iT_s)^{-1}}\norm{\Bar{B}(iT_s)} \norm{\sigma(t) - \sigma(\tau_j)},\vspace{-3mm}
\end{equation*}
and
\begin{align*}
\Phi_2(t) \triangleq  &\left(\norm{\Bar{B}(iT_s)-\bar{B}(\tau_j)} \right. \nonumber \left. +\norm{(I-\exp(A_sT_s))\bar{B}(\tau_j)} \right) \nonumber 
\\
&\times \sqrt{6}\norm{\sigma(\tau_j)}\norm{\Bar{B}(iT_s)^{-1}}\nonumber,
\end{align*}
with $\tau_j \in [(i-1)T_s,iT_s)$, and $i \geq 1$.

To compute the estimation error bound, we first compute the bound on $\Phi_1(t)$ in~\eqref{estimation error bound derivation}. The results from Proposition~\ref{proposition: bound on x dot} imply that for $\tau_j \in [(i-1)T_s,iT_s)$ and $t\in [iT_s,(i+1)T_s)$
\begin{equation*}
\norm{x(t)-x(\tau_j)} \leq \textstyle \int_{\tau^\ast}^t \norm{\dot{x}(\tau)} \dd\tau \leq 2T_s \phi_1.
\end{equation*}
Additionally, since $\sigma(t,x(t))$ is Lipschitz due to Assumption~\ref{aspt: main assumption}, using the bounds of $\norm{\Bar{B}}$ and $\norm{\Bar{B}^{-1}}$ in~\eqref{bounds on Bs}, we have
\begin{equation}\label{bound of sigma}
\begin{aligned}
    \norm{\Phi_1} \leq 2\sqrt{6}T_s\Delta_{\Bar{B}}\Delta_{\Bar{B}^{-1}}\left( \phi_1L_{\sigma_x} +  L_{\sigma_t}\right),
\end{aligned}
\end{equation}
where $L_{\sigma_x}$ and $L_{\sigma_t}$ are the Lipschitz constants of $\sigma(t,x(t))$ with respect to the state $x$ and time $t$, respectively. 
Next, we derive the bound on $\Phi_2$. Since we assumed $x(t) \in \mathcal{O}(x_d(t),\rho)$, $\norm{\Omega}$ is bounded. Thus, $\dot{R}$ is bounded because $\dot{R} = R \Omega^{\wedge}$ from~\eqref{eq:angle integral}. Notice that $\Bar{B}(t)=\bar{B}(R(t))=\left[ B(R(t)) \ B^{\bot}(R(t)) \right]=\left[  \begin{smallmatrix}
 -m^{-1}R(t)e_3 & 0_{3\times3} & m^{-1}R(t)e_1 &m^{-1}R(t)e_2\\
 0_{3\times1} & J^{-1} & 0_{3\times1} & 0_{3\times1}
\end{smallmatrix} \right]$. We then conclude that $\norm{\dot{\bar{B}}(t)}$ is bounded and define $L_B$ such that $\norm{\dot{\bar{B}}(t)} \leq L_B$. With the help of $L_B$, we have
\begin{equation}\label{bound of B bar}
    \norm{\Bar{B}(iT_s) - \Bar{B}(\tau_j)} \leq L_BT_s.
\end{equation}
Using the inequality in~\eqref{bound of B bar} and the bound of $\norm{\bar{B}^{-1}}$, $\norm{\bar{B}}$ in~\eqref{bounds on Bs}, we have
\begin{equation}\label{equ: bound of phi5}
    \begin{aligned}
       \norm{\Phi_2}  \leq&  \sqrt{6} \Delta_{\Bar{B}^{-1}}\left(  L_B T_s + \norm{I-\exp(A_sT_s)}  \Delta_{\Bar{B}}\right) \Delta_{\sigma} \\ \leq&  \sqrt{6} \Delta_{\Bar{B}^{-1}}T_s\left(  2L_B + \left| \lambda_{m}(A_s) \right|  \Delta_{\Bar{B}}\right) \Delta_\sigma.
    \end{aligned}
\end{equation}
Plugging in~\eqref{bound of sigma} and ~\eqref{equ: bound of phi5} into~\eqref{estimation error bound derivation}, we obtain
\begin{align}
     \norm{\sigma(t) - \hat{\sigma}(t)} \nonumber \leq& \sqrt{6} \Delta_{\Bar{B}^{-1}}T_s\left(  2L_B + \left| \lambda_{m}(A_s) \right|  \Delta_{\Bar{B}}\right) \Delta_\sigma \nonumber \\ &+2\sqrt{6}T_s\Delta_{\Bar{B}}\Delta_{\Bar{B}^{-1}}\left( \phi_1L_{\sigma_x} +  L_{\sigma_t}\right) \nonumber \\
    =& \zeta_2(A_s)T_s,
\end{align}
for $t \in [iT_s,(i+1)T_s)$, and $i \geq 1$ where $\zeta_2(A_s)$ is defined in~\eqref{equ: zeta 2}. Note that since $i \geq 1$, the bound $\norm{\sigma(t)-\hat{\sigma}(t)} \leq \zeta_2(A_s)T_s$ holds for any $t \geq T_s$. When $t = 0$, the prediction error $\tilde{z}(0)=0$ from~\eqref{equ:prediction error dynamics}, which leads to $\hat{\sigma}(t)=0$ for $t \in [0,T_s)$ according to~\eqref{equ: sigma hat}. Thus, the estimation error is bounded as $\norm{\sigma(t) - \hat{\sigma}(t)}=\norm{\sigma(t)}\leq \Delta_\sigma$ for $t \in [0,T_s)$ by Assumption~\ref{aspt: main assumption}.
\end{proof}
\section{Proof of Theorem~\ref{thm: main theorem}}\label{apdx: proof of thm}
We begin by introducing the error dynamics between the uncertain system in~\eqref{equ:quadrotor partial states uncertain dynamics} and the nominal dynamics in~\eqref{equ:quadrotor dynamics}. Then, we compute the Lie derivative of the Lyapunov function in~\eqref{equ:lyapunovforcomplete}. Finally, we prove this theorem by contradiction.

Decomposing the state-space representation of the uncertain partial dynamics in~\eqref{equ:quadrotor partial states uncertain dynamics}, we have
\begin{subequations}\label{equ: uncertain partial dynamics in differential equation}
\begin{align}\label{equ: uncertain partial dynamics}
    m \dot{v} &= mge_3 -fRe_3 - \sigma_f Re_3 + [Re_1 \ Re_2]\sigma_{um}, \\
    J \dot{\Omega} &= M + \sigma_M - \Omega \times J\Omega, \label{equ:uncertainrotationaldynamics}
\end{align}
\end{subequations}
where the matched uncertainties $\sigma_m = [\sigma_f \ \sigma_M^{\top}]^{\top}$ are partitioned into components on the thrust $\sigma_f \in \mathbb{R}$ and on the moment $\sigma_M \in \mathbb{R}^3$. The control inputs $f$ and $M$ consist of baseline control and adaptive control such that $f=f_b+f_{ad}$ and $M=M_b+M_{ad}$.

Now we compute $m\dot{e}_v$ and $J\dot{e}_\Omega$ between the nominal dynamics in~\eqref{equ:quadrotor dynamics} and the uncertain dynamics~\eqref{equ: uncertain partial dynamics in differential equation}. We begin by considering the  translational error dynamics:
\begin{align*}
m \dot{e}_v =& m\dot{v} - m\dot{v}_d \\ =& mge_3 -fRe_3 - \sigma_fRe_3 + [Re_1 \ Re_2]\sigma_{um}-m \Ddot{x}_d.
\end{align*}
Then we plug in $f=f_b+f_{ad}$, and add and subtract $\frac{f_b}{e_3^{\top}R_d^{\top}Re_3}R_de_3$ to get \vspace{-2mm}
\begin{align*}
m\dot{e}_v =& mge_3 - m\Ddot{x}_d-\frac{f_b}{e_3^{\top}R_d^{\top}Re_3}R_de_3-X \\ &- (f_{ad}+\sigma_f)Re_3 + [Re_1 \ Re_2]\sigma_{um},
\end{align*}
where 
$X=\left(  (e_3^{\top} R_d^{\top} R e_3)Re_3 -R_de_3 \right) \left(f_b / (e_3^{\top} R_d^{\top} R e_3)\right)$.
By~\eqref{equ:translational dynamics} and the following equality in~\cite[After eq. (77)]{lee2010arxiv}\vspace{-1mm}
\[
-\frac{f_b}{e_3^{\top}R_d^{\top}Re_3}R_de_3 = -K_x e_x -K_ve_v-mge_3 + m\Ddot{x}_d,
\]
we have the time-derivative of the velocity error:
\begin{align}\label{equ:timederivativeofvelocityerror}
    m\dot{e}_v =& -K_x e_x -K_ve_v-X -(f_{ad}+\sigma_f)Re_3 \nonumber \\&+ [Re_1 \ Re_2]\sigma_{um}.
\end{align}
The rotational error dynamics are given by~\cite[After eq. (55)]{lee2010arxiv} such that
$\dot{e}_\Omega = \dot{\Omega} + \hat{\Omega}R^{\top} R^{\top}_d \Omega_d - R^{\top}R_d \dot{\Omega}_d$.
Multiplying $J$ on both sides and substituting $J\dot{\Omega}$ from~\eqref{equ:uncertainrotationaldynamics}, 
we have\vspace{-1mm}
\begin{align}\label{equ:timederivativeofangularrateerror}
    J\dot{e}_\Omega =& -K_Re_R-K_\Omega e_\Omega +\Omega \times J\Omega \nonumber\\&-J(\hat{\Omega} R^{\top}R_d\Omega_d-R^{\top}R_d\dot{\Omega}_d) + M_{ad} + \sigma_M \nonumber\\&- \Omega \times J\Omega + J(\hat{\Omega} R^{\top}R_d\Omega_d-R^{\top}R_d\dot{\Omega}_d) \nonumber \\
    =& -K_Re_R-K_\Omega e_\Omega + M_{ad} +\sigma_M,
\end{align}
where we use the baseline control input $M_b$ from~\eqref{eq:torque control}.

We are now ready to compute the derivative of the Lyapunov function in~\eqref{equ:lyapunovforcomplete}:
\begin{align}\label{equ:timederivativeoflyapunovfunction}
    \dot{V} = &K_x e_x \cdot e_v  + c_1e_v \cdot e_v + \frac{c_1}{m}e_x \cdot m\dot{e}_v + e_\Omega \cdot J \dot{e}_\Omega\nonumber \\& + e_v \cdot m\dot{e}_v+ K_R e_R \cdot e_\Omega + c_2 \dot{e}_R \cdot e_\Omega +c_2 e_R \cdot \dot{e}_\Omega.
\end{align}
Substituting~\eqref{equ:timederivativeofvelocityerror} and~\eqref{equ:timederivativeofangularrateerror} into~\eqref{equ:timederivativeoflyapunovfunction}, we have
\begin{align*}
    \dot{V}=&-(K_v -c_1) \norm{e_v}^2 -\frac{c_1 K_x}{m}\norm{e_x}^2 - \frac{c_1 K_v}{m}e_x \cdot e_v \\&- X \cdot \left(  \frac{c_1}{m}e_x + e_v \right) -(f_{ad}+\sigma_f)Re_3 \cdot \left( \frac{c_1}{m}e_x + e_v  \right) \\&+ ([Re_1 \ Re_2]\sigma_{um})\cdot \left( \frac{c_1}{m}e_x+e_v  \right)- K_\Omega \norm{e_\Omega}^2 \\ & -c_2K_Re_R \cdot J^{-1}e_R + c_2C(R_d^{\top}R)e_\Omega \cdot e_\Omega\\& -c_2K_\Omega e_R\cdot J^{-1}e_\Omega+ (e_\Omega^\top+ c_2e_R^\top J^{-1})(M_{ad}+\sigma_M).
\end{align*}
By~\eqref{equ:lyapunovdotsmallerthanlyapunov}, we have 
\begin{align}\label{equ:lyapunovdotsmallthanlyapunov2}
   \dot{V} \leq& \begin{bmatrix}
-R(t)e_3 \cdot \left( \frac{c_1}{m}e_x(t) + e_v(t)  \right) \\ e_\Omega(t)+c_2 J^{-1} e_R(t) \end{bmatrix}^\top
   \begin{bmatrix}
   f_{ad}+\sigma_f \\ M_{ad}+\sigma_M\end{bmatrix} \nonumber \\ &-\beta V +\left([Re_1 \ Re_2]\sigma_{um}\right)\cdot \left( \frac{c_1}{m}e_x+e_v  \right).
\end{align}
Define $\eta_m(s) \triangleq C(s) \mathcal{L}\left[ \sigma_m(t,x(t)) \right]$, and $\hat{\eta}_m(s) \triangleq C(s) \mathcal{L}\left[ \hat{\sigma}_m(t,x(t)) \right]$. Adding and subtracting $P_1\eta_m$ in~\eqref{equ:lyapunovdotsmallthanlyapunov2}, using $u_{ad}=[f_{ad} \ M_{ad}^\top]^\top$ and~\eqref{equ:L1controlinput},  we have
\begin{align}\label{equ: v dot final}
    \dot{V} \leq& -\beta V + P_1(\sigma_m - \eta_m) + P_1(\eta_m - \hat{\eta}_m) \nonumber \\ &+ ([Re_1 \ Re_2]\sigma_{um}) \cdot P_2,
\end{align}
where $P_1$ and $P_2$ are defined in~\eqref{def: P1 P2 definition}. 

We are now ready to prove Theorem~\ref{thm: main theorem}.

\begin{proof}
We prove this Theorem by contradiction: assume that $d(x(t),x_d(t)) \geq \rho$ for some $t > 0$. Since $d(x(0),x_d(0)) < \rho$, there exists $\tau^{\ast} > 0$ such that
\begin{align*}
    d(x(t),x_d(t)) <& \rho, \quad t\in[0,\tau^{\ast}). \\
    d(x(\tau^\ast),x_d(\tau^\ast)) =& \rho. 
\end{align*}
Following the differential inequality in~\eqref{equ: v dot final}, we have
\begin{equation}\label{equ:finallyapunovinL1proof}
        V(t) \leq e^{-\beta t} V_0 + \Phi_3(t) + \Phi_4(t) + \Phi_5(t),
\end{equation}
where
\begin{align*}
\Phi_3(t) &\triangleq \textstyle \int ^{t}_0 e^{-\beta(t-\tau)}P_1(\tau)(\sigma_m(\tau)-\eta_m(\tau))\dd\tau, \\
\Phi_4(t) &\triangleq \textstyle \int ^{t}_0 e^{-\beta(t-\tau)}P_1(\tau)\tilde{\eta}_m(\tau)\dd\tau, \\
\Phi_5(t) &\triangleq \textstyle \int ^{t}_0 e^{-\beta(t-\tau)}([Re_1 \ Re_2]\sigma_{um}) \cdot P_2 \dd\tau,
\end{align*}
and $\tilde{\eta}_m(t)=\eta_m(t)-\hat{\eta}_m(t)$.
Using the bounds in~\eqref{bounds of P1 P2}, we first derive the bound on $\Phi_3$. The integral $\Phi_3$ can be expressed as the solution to the following virtual scalar system
\begin{subequations}\label{equ:scalarlinearsystem}
\begin{align}\label{equ:scalarlinearsystem1}
    \dot{y}(t) =& -2\beta y(t)+ P_1(t)\xi(t), \quad y(0) = 0, \\
\label{equ:scalarlinearsystem2}
    \xi(s)=&(1-C(s))\mathcal{L}\left[\sigma_m(t,x(t)\right].
\end{align}
\end{subequations}
From Assumption~\ref{aspt: main assumption} and Proposition~\ref{proposition: bound on x dot}, the following bound holds for all $x(t) \in \mathcal{O}(x_d(t),\rho)$ and all $t \in [0,\tau^\ast]$: 
\begin{align*}
\norm{\dot{\sigma}_m(t,x(t))}=&\norm{\frac{\partial \sigma_m(t,x(t))}{\partial t} + \frac{\partial \sigma_m(t,x(t))}{\partial x}\dot{x}(t)} \\ \leq& L_{\sigma_{m_t}}+L_{\sigma_{m_x}}\phi_1.
\end{align*}
By~\cite[ Lemma A.1]{lakshmanan2020safe}, the solution of the linear system~\eqref{equ:scalarlinearsystem} satisfies the following norm bound:
\[
\norm{y(t)} \leq \norm{P_1} \left(\frac{\norm{\sigma_m(0)}}{\abs{\beta - \omega}} + \frac{\norm{\dot{\sigma}_m(t,x(t))}}{\beta\omega}\right).
\]
Since $y(t)=\Phi_3(t)$, we have
\begin{align}\label{equ:bound on phi1}
\norm{\Phi_3} \leq&  c_3 \rho \left( \frac{\Delta_{\sigma_m}}{\lvert \beta - \omega \rvert} + \frac{L_{\sigma_{m_t}} + L_{\sigma_{m_x}}\phi_1}{\beta \omega} \right)= c_3 \rho \zeta_1(\omega).
\end{align}
Note that we consider a first-order low-pass filter of the form $C(s)=\omega/ \left(s+\omega\right)$. The results can be generalized to higher-order low-pass filters. 
We now derive the bound on $\Phi_4$. As an intermediate step, we will derive $\norm{\tilde{\eta}}^{[0,T_s)}_{\mathcal{L}_{\infty}}$ and $\norm{\tilde{\eta}}^{[T_s,\tau^\ast]}_{\mathcal{L}_{\infty}}$ individually, where $\tilde{\eta}(s)=C(s)\left( \mathcal{L}[\sigma(t,x(t))]-\mathcal{L}[\hat{\sigma}(t,x(t))] \right)$.
We first show the bound of $\norm{\tilde{\eta}(t)}$ for $t \in [0,T_s)$. Since the prediction error at $t=0$ is $0$, the uncertainty estimate $\hat{\sigma}(t)$ for $t \in [0,T_s)$ is $0$. As a result, $\tilde{\eta}(s) = (\sigma(s)-0) \omega / \left(s+\omega\right)$. Take the inverse Laplace transform, and we get
\[
\dot{\tilde{\eta}}(t)=-\omega \tilde{\eta}(t) + \omega \sigma(t).
\]
Solving the above differential equation, we have
\begin{equation}
\begin{aligned}
   \tilde{\eta}(t) & =\textstyle \int_0^{t} \exp(-w(t-\tau))\omega\sigma(\tau)\dd\tau 
   \\&\leq \textstyle \int_0^{t} \exp(-w(t-\tau))\dd\tau \omega\Delta_{\sigma} \\
   & = \Delta_{\sigma}(1-\exp(-\omega t))  \leq \Delta_\sigma \omega t = \zeta_3(\omega)t.
\end{aligned}
\end{equation}
As a result, we get\begin{equation}\label{equ: eta tilde bound1}
    \norm{\tilde{\eta}}_{\mathcal{L}_\infty}^{[0,T_s]} \leq \zeta_3(\omega)T_s,
\end{equation}
and
\begin{equation}\label{equ: bound on eta tilde 1}
    \norm{\tilde{\eta}_m}_{\mathcal{L}_\infty}^{[0,T_s]} \leq \norm{\tilde{\eta}}_{\mathcal{L}_\infty}^{[0,T_s]} \leq   \zeta_3(\omega)T_s.
\end{equation}
Next we show the bound for $\norm{\tilde{\eta}(t)}$ for $t \geq T_s$. Following the same idea as above, we obtain
\[
\dot{\tilde{\eta}}(t) = -\omega \tilde{\eta}(t) + \omega (\sigma(t) - \hat{\sigma}(t)),
\]
for $t \geq T_s$. Solving the above differential equation, we have
\begin{align*}
  \tilde{\eta}(t)=&\exp(-\omega(t-T_s))\tilde{\eta}(T_s)  \\ &+\textstyle \int_{T_s}^{t} \exp(-\omega(t-\tau))\omega (\sigma(\tau)-\hat{\sigma}(\tau)) \dd \tau.
\end{align*}
Using the result from Proposition~\ref{prop: estimation error bound}, where $\norm{\sigma(t)-\hat{\sigma}(t)}_{\mathcal{L}_\infty}^{[T_s,t]} \leq \zeta_2(A_s)T_s$ and~\eqref{equ: eta tilde bound1}, we have
\begin{align*}
\norm{\tilde{\eta}(t)}_{\mathcal{L}_{\infty}}^{[T_s,t]} \leq& \norm{\exp(-\omega(t-T_s))}\zeta_3(\omega)T_s \\ &+ \norm{(1-\exp(-\omega(t-T_s)))}\zeta_2(A_s)T_s.
\end{align*}
Since $\omega$ is positive, it is straightforward to conclude $\norm{\tilde{\eta}}_{\mathcal{L}_{\infty}}^{[T_s,t]} \leq (\zeta_2(A_s) + \zeta_3(\omega))T_s$. In addition, since $\norm{\tilde{\eta}_m}^{[T_s,t]}_{\mathcal{L}_{\infty}} \leq \norm{\tilde{\eta}}^{[T_s,t]}_{\mathcal{L}_{\infty}}$, we obtain
\begin{equation}\label{equ: bound on eta tilde 2}
\norm{\tilde{\eta}_m}^{[T_s,t]}_{\mathcal{L}_{\infty}} \leq (\zeta_2(A_s) + \zeta_3(\omega))T_s.
\end{equation}
Without loss of generality, considering the case of $t > T_s$, with the bound on $\norm{\Tilde{\eta}_m(t)}$ from~\eqref{equ: bound on eta tilde 1},~\eqref{equ: bound on eta tilde 2} and the fact that $\norm{P_1} \leq c_3 \rho$, we can show that
\begin{align}\label{equ:bound on phi2}
 \norm{\Phi_4} \leq& c_3 \rho\left(\textstyle\int_0^{T_s}e^{-\beta(T_s-\tau)}\norm{\tilde{\eta}_m(\tau)}\dd\tau \right. \nonumber \\ &+ \left. \textstyle\int_{T_s}^{t}e^{-\beta(t-\tau)}\norm{\tilde{\eta}_m(\tau)}\dd\tau\right) \nonumber \\
  \leq& \frac{c_3 \rho}{\beta} (\zeta_3(\omega)+\zeta_2(A_s))T_s=\zeta_4(A_s,w)T_s.
\end{align}
We now derive the bound on $\Phi_5$. Since $\norm{Re_1}=\norm{Re_2}=1$ by the definition of the rotation matrix and $\norm{P_2} \leq  c_4 \rho$, we have
\begin{equation}\label{equ:bound on Phi 3}
    \norm{\Phi_5} \leq  c_4 \rho \Delta_{\sigma_{um}}.
\end{equation}
Finally, substituting in the results from~\eqref{equ:bound on phi1},~\eqref{equ:bound on phi2}, and~\eqref{equ:bound on Phi 3} into~\eqref{equ:finallyapunovinL1proof}, we obtain
\begin{equation}\label{equ:lyapunov for ultimate bound}
V \leq e^{-\beta t}V_0 + c_3 \rho \zeta_1(\omega) + \zeta_4(A_s,\omega)T_s + c_4\rho\Delta_{\sigma_{um}}
\end{equation}
for all $t \in [0,\tau^{\ast}]$. Recall from the assumption that $d\left(x(\tau^{\ast}),x_d(\tau^{\ast})\right) = \rho$. Using~\eqref{equ: lyapunovboundfinal}, we have
\[
\underline{\gamma}\rho^2 < V_0 +  c_3 \rho \zeta_1(\omega) + \zeta_4(A_s,\omega)T_s + c_4\rho\Delta_{\sigma_{um}}.
\]
Moving $T_s$ to the left side, we have
\[
T_s > \frac{\underline{\gamma}\rho^2- c_3 \rho \zeta_1(\omega) - c_4\rho\Delta_{\sigma_{um}}-V_0}{\zeta_4(A_s,\omega)},
\]
which contradicts~\eqref{equ: condition on sampling rate}. Therefore, we conclude $d(x(t),x_d(t))<\rho$ for all $t$. Moreover, since $d(x(t),x_d(t))<\rho$, using~\eqref{equ:lyapunov for ultimate bound} and~\eqref{equ: lyapunovboundfinal}, we can obtain the following uniform ultimate bound, i.e., for any $t_1 > 0$
\begin{equation}
\begin{aligned}
d&(x(t),x_d(t)) < \mu(\omega,T_s,t_1) \triangleq \\
&\sqrt{\frac{e^{-\beta t_1}V_0 + c_3 \rho \zeta_1(\omega) + \zeta_4(A_s,\omega)T_s + c_4\rho\Delta_{\sigma_{um}}}{\underline{\gamma}}}
\end{aligned}
\end{equation}
holds for all $t\geq t_1$.
\end{proof}

\section{RMSEs in the benchmark experiments}\label{apdx: benchmark data}
We supplement the RMSEs in the benchmark experiments in Tables~\ref{tb: added weight benchmark RMSE circle}--\ref{tb: slung weight benchmark RMSE figure 8} (on the next page).
\setlength{\tabcolsep}{2pt} 
\renewcommand{\arraystretch}{1} 
  \captionsetup{
	skip=5pt, position = bottom}
\begin{table*}[h]
	\centering
	\small
	\caption{RMSEs in the added weight benchmark with the circular trajectory.}
\begin{tabular}{ccccccccccccccccccccccccc}
		\toprule[1pt]
		 Speed [m/s] && \multicolumn{3}{c}{0} && \multicolumn{3}{c}{0.5} && \multicolumn{3}{c}{1} && \multicolumn{3}{c}{1.5} && \multicolumn{3}{c}{2} && \multicolumn{3}{c}{2.5}\\
		\midrule
		Weight [g] &&  $\mathcal{L}_1$off &  $\mathcal{L}_1$on & PID &&  $\mathcal{L}_1$off &  $\mathcal{L}_1$on & PID &&  $\mathcal{L}_1$off &  $\mathcal{L}_1$on & PID &&  $\mathcal{L}_1$off &  $\mathcal{L}_1$on & PID &&  $\mathcal{L}_1$off &  $\mathcal{L}_1$on & PID && $\mathcal{L}_1$off &  $\mathcal{L}_1$on & PID \\
		\cmidrule{1-1} \cmidrule{3-5} \cmidrule{7-9} \cmidrule{11-13} \cmidrule{15-17} \cmidrule{19-21} \cmidrule{23-25}
        0 && 0.08 & 0.03 & 0.01 && 0.08 & 0.04 & 0.02 &&0.09 & 0.03 & 0.03 &&0.11 & 0.05 & 0.05 &&0.18 & 0.11 & 0.07 &&0.24 & 0.15 & 0.09 \\ 
100 && 0.05 & 0.03 & 0.01 && 0.06 & 0.04 & 0.02 &&0.06 & 0.03 & 0.03 &&0.09 & 0.05 & 0.06 &&0.16 & 0.12 & 0.10 &&0.24 & 0.18 & 0.13 \\ 
200 && 0.06 & 0.03 & 0.03 && 0.07 & 0.04 & 0.03 &&0.07 & 0.04 & 0.04 &&0.10 & 0.05 & 0.07 &&0.17 & 0.11 & 0.11 &&0.22 & 0.16 & 0.12 \\ 
300 && 0.12 & 0.03 & 0.01 && 0.12 & 0.03 & 0.03 &&0.14 & 0.02 & 0.03 &&0.16 & 0.05 & 0.06 &&0.24 & 0.11 & 0.12 &&0.33 & 0.15 & 0.15 \\ 
400 && 0.16 & 0.03 & 0.01 && 0.16 & 0.04 & 0.03 &&0.10 & 0.03 & 0.03 &&0.19 & 0.05 & 0.08 &&0.28 & 0.12 & 0.11 &&0.35 & 0.16 & 0.19 \\ 
500 && 0.18 & 0.03 & 0.01 && 0.17 & 0.04 & 0.03 &&0.16 & 0.03 & 0.02 &&0.17 & 0.05 & 0.05 &&0.34 & 0.12 & 0.13 &&0.43 & 0.17 & 0.16 \\ 
		\bottomrule[1pt]
	\end{tabular}\label{tb: added weight benchmark RMSE circle}
\end{table*}
\normalsize

\setlength{\tabcolsep}{2pt} 
\renewcommand{\arraystretch}{1} 
  \captionsetup{
	skip=5pt, position = bottom}
\begin{table*}[h]
	\centering
	\small
	\caption{RMSEs in the slung weight benchmark with the circular trajectory. The strikethrough data indicate crash in the end.}
\begin{tabular}{ccccccccccccccccccccccccc}
		\toprule[1pt]
		 Speed [m/s] && \multicolumn{3}{c}{0} && \multicolumn{3}{c}{0.5} && \multicolumn{3}{c}{1} && \multicolumn{3}{c}{1.5} && \multicolumn{3}{c}{2} && \multicolumn{3}{c}{2.5}\\
		\midrule
		Weight [g] &&  $\mathcal{L}_1$off &  $\mathcal{L}_1$on & PID &&  $\mathcal{L}_1$off &  $\mathcal{L}_1$on & PID &&  $\mathcal{L}_1$off &  $\mathcal{L}_1$on & PID &&  $\mathcal{L}_1$off &  $\mathcal{L}_1$on & PID &&  $\mathcal{L}_1$off &  $\mathcal{L}_1$on & PID && $\mathcal{L}_1$off &  $\mathcal{L}_1$on & PID \\
		\cmidrule{1-1} \cmidrule{3-5} \cmidrule{7-9} \cmidrule{11-13} \cmidrule{15-17} \cmidrule{19-21} \cmidrule{23-25}
        0 && 0.08 & 0.03 & 0.01 && 0.08 & 0.04 & 0.02 &&0.09 & 0.03 & 0.03 &&0.11 & 0.05 & 0.05 &&0.18 & 0.11 & 0.07 &&0.24 & 0.15 & 0.09 \\ 
100 && 0.04 & 0.03 & 0.01 && 0.05 & 0.03 & 0.02 &&0.06 & 0.03 & 0.03 &&0.09 & 0.05 & 0.06 &&0.17 & 0.12 & 0.07 &&0.23 & 0.16 & \st{0.12} \\ 
200 && 0.07 & 0.03 & 0.01 && 0.08 & 0.03 & 0.02 &&0.07 & 0.02 & 0.03 &&0.12 & 0.07 & 0.06 &&0.19 & 0.12 & \st{0.14} &&0.25 & 0.16 & \st{0.14} \\ 
300 && 0.11 & 0.03 & 0.01 && 0.12 & 0.03 & 0.02 &&0.11 & 0.02 & 0.03 &&0.16 & 0.08 & 0.06 &&0.24 & 0.13 & \st{0.13} &&0.33 & 0.17 & \st{0.18} \\ 
400 && 0.15 & 0.03 & 0.01 && 0.15 & 0.03 & 0.02 &&0.13 & 0.02 & 0.04 &&0.15 & 0.05 & 0.07 &&0.29 & 0.14 & \st{0.16} &&0.40 & 0.18 & \st{0.30} \\ 
500 && 0.20 & 0.03 & 0.02 && 0.19 & 0.03 & 0.02 &&0.18 & 0.03 & 0.06 &&0.18 & 0.05 & 0.10 &&0.38 & 0.13 & \st{0.20} &&0.56 & 0.17 & \st{0.36} \\ 
		\bottomrule[1pt]
	\end{tabular}\label{tb: slung weight benchmark RMSE circle}
\end{table*}
\normalsize

\setlength{\tabcolsep}{2pt} 
\renewcommand{\arraystretch}{1} 
  \captionsetup{
	skip=5pt, position = bottom}
\begin{table*}[h]
	\centering
	\small
	\caption{RMSEs in the added weight benchmark with figure 8 trajectory.}
\begin{tabular}{ccccccccccccccccc}
		\toprule[1pt]
		 Speed [m/s] && \multicolumn{3}{c}{0} && \multicolumn{3}{c}{1} && \multicolumn{3}{c}{2} && \multicolumn{3}{c}{3} \\
		\midrule
		Weight [g] &&  $\mathcal{L}_1$off &  $\mathcal{L}_1$on & PID &&  $\mathcal{L}_1$off &  $\mathcal{L}_1$on & PID &&  $\mathcal{L}_1$off &  $\mathcal{L}_1$on & PID &&  $\mathcal{L}_1$off &  $\mathcal{L}_1$on & PID  \\
		\cmidrule{1-1} \cmidrule{3-5} \cmidrule{7-9} \cmidrule{11-13} \cmidrule{15-17} 
        0 && 0.08 & 0.03 & 0.01 && 0.08 & 0.04 & 0.02 &&0.09 & 0.05 & 0.03 &&0.11 & 0.07 & 0.06 \\ 
100 && 0.05 & 0.03 & 0.01 && 0.06 & 0.04 & 0.02 &&0.07 & 0.04 & 0.03 &&0.09 & 0.06 & 0.06 \\ 
200 && 0.06 & 0.03 & 0.01 && 0.10 & 0.04 & 0.01 &&0.10 & 0.05 & 0.05 &&0.10 & 0.07 & 0.08 \\ 
300 && 0.12 & 0.03 & 0.01 && 0.13 & 0.03 & 0.02 &&0.13 & 0.04 & 0.04 &&0.13 & 0.07 & 0.08 \\ 
400 && 0.16 & 0.04 & 0.01 && 0.12 & 0.04 & 0.02 &&0.17 & 0.05 & 0.05 &&0.16 & 0.07 & 0.10 \\ 
500 && 0.19 & 0.03 & 0.01 && 0.17 & 0.04 & 0.03 &&0.22 & 0.04 & 0.05 &&0.22 & 0.07 & 0.11 \\  
		\bottomrule[1pt]
	\end{tabular}\label{tb: added weight benchmark RMSE figure 8}
\end{table*}
\normalsize

\setlength{\tabcolsep}{2pt} 
\renewcommand{\arraystretch}{1} 
  \captionsetup{
	skip=5pt, position = bottom}
\begin{table*}[h]
	\centering
	\small
	\caption{RMSEs in the slung weight benchmark with figure 8 trajectory.}
\begin{tabular}{ccccccccccccccccc}
		\toprule[1pt]
		 Speed [m/s] && \multicolumn{3}{c}{0} && \multicolumn{3}{c}{1} && \multicolumn{3}{c}{2} && \multicolumn{3}{c}{3} \\
		\midrule
		Weight [g] &&  $\mathcal{L}_1$off &  $\mathcal{L}_1$on & PID &&  $\mathcal{L}_1$off &  $\mathcal{L}_1$on & PID &&  $\mathcal{L}_1$off &  $\mathcal{L}_1$on & PID &&  $\mathcal{L}_1$off &  $\mathcal{L}_1$on & PID  \\
		\cmidrule{1-1} \cmidrule{3-5} \cmidrule{7-9} \cmidrule{11-13} \cmidrule{15-17} 
        0 && 0.08 & 0.03 & 0.01 && 0.08 & 0.04 & 0.02 &&0.09 & 0.05 & 0.03 &&0.11 & 0.07 & 0.06 \\ 
100 && 0.04 & 0.03 & 0.01 && 0.06 & 0.04 & 0.02 &&0.07 & 0.04 & 0.04 &&0.09 & 0.07 & 0.07 \\ 
200 && 0.07 & 0.03 & 0.02 && 0.07 & 0.04 & 0.02 &&0.07 & 0.05 & 0.04 &&0.10 & 0.07 & 0.08 \\ 
300 && 0.11 & 0.03 & 0.01 && 0.08 & 0.04 & 0.04 &&0.12 & 0.05 & 0.04 &&0.13 & 0.08 & 0.09 \\ 
400 && 0.15 & 0.03 & 0.02 && 0.13 & 0.04 & 0.03 &&0.14 & 0.04 & 0.04 &&0.19 & 0.07 & 0.09 \\ 
500 && 0.20 & 0.03 & 0.01 && 0.17 & 0.04 & 0.02 &&0.18 & 0.04 & 0.04 &&0.22 & 0.07 & 0.12 \\ 
		\bottomrule[1pt]
	\end{tabular}\label{tb: slung weight benchmark RMSE figure 8}
\end{table*}
\normalsize
\end{document}